\documentclass{llncs}
\usepackage{textcomp}
\usepackage{float}
\usepackage{listings}
\lstdefinelanguage{pseudo}{morekeywords={init,with,or,if,then,else,fi,and,not,while,do,od,distinct,
    case, goto,local,algorithm, function, for, each, times, from, to,
    variables, procedure, recursive, return},
  morecomment=[l]{//}, morecomment=[s]{/*}{*/},
  mathescape=true,escapechar={@},
  basicstyle=\sffamily\small,
  commentstyle=\itshape\rmfamily\small,
  keywordstyle=\sffamily\bfseries\small
}
\usepackage{mathtools}
\usepackage{amsfonts,amssymb,amsmath}

\usepackage{etoolbox}  
\usepackage{enumitem} 

\usepackage{hyperref}
\usepackage{environ}
\usepackage{csvsimple}
\usepackage{longtable} 
\usepackage[normalem]{ulem} 

\usepackage{fourier-orns} 
\usepackage{wasysym} 
\usepackage{marvosym} 
\usepackage{bbding} 

\usepackage{etoolbox}  
\usepackage{enumitem} 

\usepackage{enumerate}
\usepackage{mathrsfs}
\usepackage{mathtools}
\usepackage{appendix}
\usepackage{comment}
\usepackage{tikz}
\usepackage{array}
\usetikzlibrary{calc}
\usepackage{multicol}

\usetikzlibrary{automata,positioning}
\usetikzlibrary{fit,backgrounds,shapes,decorations} 
\definecolor{processblue}{cmyk}{0.96,0,0,0}

\newcommand{\be}{\begin{enumerate}}
\newcommand{\ee}{\end{enumerate}}
\newcommand{\bc}{\begin{center}}
\newcommand{\ec}{\end{center}}

\newcommand{\bi}{\begin{itemize}}
\newcommand{\ei}{\end{itemize}}

\usepackage[textsize=scriptsize]{todonotes}
\usepackage{hyperref}

\newcommand{\NatZero}{{\mathbb N}_0}

\newcommand{\Rational}{{\mathbb Q}}

\newcommand{\TA}{\textsf{TA}}
\newcommand{\local}{{\mathcal L}}
\newcommand{\initlocal}{{\mathcal I}}
\newcommand{\paraset}{\Pi}
\newcommand{\globset}{\Gamma}
\newcommand{\varset}{\mathcal{V}} 
\newcommand{\ruleset}{{\mathcal R}}

\newcommand{\ResCond}{{\mathit{RC}}} 
\newcommand{\syssize}{N}
\newcommand{\numlocal}{{|\local|}}

\newcommand{\numglob}{{|\globset|}}
\newcommand{\Env}{{\mathit{Env}}}

\newcommand{\update}{\vec{u}}
\newcommand{\zerovec}{\vec{0}}
\newcommand{\fromstate}{{\mathit{from}}}
\newcommand{\tostate}{{\mathit{to}}}
\newcommand{\PrecondU}{\Phi^{\mathrm{rise}}}
\newcommand{\PrecondL}{\Phi^{\mathrm{fall}}}
\newcommand{\Precond}{\Phi} 

\newcommand{\precond}{\varphi} 

\newcommand{\true}{\mathit{true}}
\newcommand{\false}{\mathit{false}}

\newcommand{\Sys}{\textsf{Sys}}
\newcommand{\configs}{\Sigma}
\newcommand{\iconfigs}{I}
\newcommand{\transrel}{T}

\newcommand\gst{\sigma}
\newcommand{\counters}{{\vec{\boldsymbol\kappa}}}
\newcommand{\vars}{\vec{g}}
\newcommand{\param}{\vec{p}}
\newcommand{\finpath}[2]{\textsf{path}(#1, #2)} 
\newcommand{\infpath}[2]{\textsf{path}(#1, #2)} 
\newcommand{\setconf}[2]{\textsf{Cfgs}(#1,#2)}   

\newcommand{\statectx}{\omega}

\newcommand\slice[2]{#1{\raise-.5ex\hbox{\ensuremath|}}_{#2}}

\newcommand\ltlF{\textsf{\textbf{F}}\,}
\newcommand\ltlG{\textsf{\textbf{G}}\,}

\newcommand\ELTLFT{\textsf{ELTL}_\textsf{FT}}

\newcommand{\casest}{\cite{CT96,ST87:abc,BrachaT85,MostefaouiMPR03,Raynal97,Gue02,DobreS06,BrasileiroGMR01,SongR08}}

\newcommand{\cpp}[1]{#1{\footnotesize{\texttt{++}}}}

\newcommand{\bala}[1]{\todo[color=blue!30]{\small #1}}

\newcommand{\ml}[1]{\todo[color=yellow,inline]{Marijana: #1}}

\begin{document}

\title{Complexity of Verification and Synthesis of Threshold Automata
\thanks{This project has received funding from the European Research Council (ERC) under the European Union's Horizon 2020 research and innovation programme under grant agreement No 787367 (PaVeS).}
}

\author{A.\,R. Balasubramanian \and 
Javier Esparza \and
Marijana Lazi\' c 
}

\authorrunning{A.\,R. Balasubramanian et al.}

\institute{Technical University of Munich, Germany}

\maketitle

\begin{abstract}
Threshold automata are a formalism for modeling and analyzing 
fault-tolerant distributed algorithms, recently introduced by Konnov, Veith, and Widder, describing protocols executed by a fixed but arbitrary number of processes. We conduct the first systematic study of the complexity of verification and synthesis problems for threshold automata. We prove that the coverability, reachability, safety, and liveness problems are NP-complete, and that the bounded synthesis problem is $\Sigma_p^2$ complete. A key to our results is a novel characterization of the reachability relation of a threshold automaton as an existential Presburger formula. The characterization also leads to novel verification and synthesis algorithms. We report on an implementation, and provide experimental results.


\end{abstract}

\section{Introduction}

Many concurrent and distributed systems consist of an 
arbitrary number of communicating processes. Parameterized verification 
investigates how to prove them correct for any number of processes~\cite{2015Bloem}. 

Parameterized systems whose processes are indistinguishable and finite state are often called 
\emph{replicated systems}. A global state of a replicated system is completely determined by the number of processes in each state. Models of replicated systems differ in the communication mechanism between processes.  Vector Addition Systems  (VAS) and their extensions \cite{GermanS92,DufourdFS98,EsparzaFM99,BlondinHM18} can model rendez-vous,
multiway synchronization, global resets and broadcasts, and other mechanisms. The decidability and complexity of their verification problems is well understood \cite{EsparzaN94,Esparza96,SchmitzS14,2015Bloem,BlondinHM18}.

Transition guards of VAS-based replicated systems are \emph{local}: Whether a transition is enabled or not depends only on the current states of a \emph{fixed} number of processes, 
independent of the total number of processes. 
Konnov \textit{et al.} observed in~\cite{KVW14:concur} that local guards cannot model fault-tolerant distributed algorithms. Indeed, in such algorithms often a process can only make a step if it has received a message from a \emph{majority} or some fraction of the processes. To remedy this, they introduced \emph{threshold automata}, a model of replicated systems with shared-variable communication and \emph{threshold guards}, in which the value of a global variable is compared to an affine combination of  the total numbers of processes of different types. 
In a number of papers, Konnov \textit{et al.}  have developed and implemented 
verification algorithms for safety and liveness of threshold automata \cite{KVW14:concur,KVW17:IandC,ELTLFT,Reach,KW18}. Further, Kukovec \textit{et al.} have obtained decidability and undecidability results \cite{AllFlavors} for different variants of the model. However, contrary to the VAS case, the computational complexity of the main verification 
problems has not yet been studied. 

We conduct the first systematic complexity analysis of threshold automata.
In the first part of the paper we show that the parameterized coverability and reachability problems
are NP-complete. Parameterized coverability asks if some configuration reachable from some initial configuration puts at least one process in a given state, and parameterized reachability asks if it puts processes in \emph{exactly} a given set of states, leaving all other states unpopulated.  The NP upper bound is a consequence of our main result, showing that the reachability relation of threshold automata is expressible in existential Presburger arithmetic.
In the second part of the paper we apply this expressibility result to prove that the model checking problem of \emph{Fault-Tolerant Temporal Logic} ($\ELTLFT$)~\cite{ELTLFT} is NP-complete, and that the problem of synthesizing the guards of a given automaton, studied in \cite{LKWB17:opodis}, is $\Sigma_p^2$ complete. 
The last part of the paper reports on an implementation of our novel approach to the parameterized (safety and liveness) verification problems. We show that it compares favorably to ByMC, the tool developed in \cite{KW18}.

Due to lack of space most proofs have been moved to the Appendix.

\section{Threshold Automata }

\begin{figure}[t]
  \begin{minipage}{.5\textwidth}
      \input{st87-pseudo.tex}
    \caption{Pseudocode of a reliable broadcast protocol from~\cite{ST87:abc}
    for a correct process~$i$, where $n$ and $t$ denote the number of processes,
    and an upper bound on the number of faulty processes. 
    The protocol satisfies its specification 
    (if $\mathit{myval}_i=1$ for every correct process $i$, then
    eventually $\mathit{accept}_j=\mathit{true}$ for some correct 
    process $j$) if $t < n/3$. }\label{fig:st} 
  \end{minipage}
  \begin{minipage}{.01\textwidth}
      \phantom{x}
  \end{minipage}
  \begin{minipage}{.49\textwidth}
  \centering 
      \tikzstyle{node}=[circle,draw=black,thick,minimum size=4.3mm,inner sep=0.75mm,font=\normalsize]
\tikzstyle{init}=[circle,draw=black!90,fill=green!10,thick,minimum size=4.3mm,inner sep=0.75mm,font=\normalsize]
\tikzstyle{final}=[circle,draw=black!90,fill=red!10,thick,minimum size=4.3mm,inner sep=0.75mm,font=\normalsize]
\tikzstyle{rule}=[->,thick]
\tikzstyle{post}=[->,thick,rounded corners,font=\normalsize]
\tikzstyle{pre}=[<-,thick]
\tikzstyle{cond}=[rounded
  corners,rectangle,minimum
  width=1cm,draw=black,fill=white,font=\normalsize]
\tikzstyle{asign}=[rectangle,minimum
  width=1cm,draw=black,fill=gray!5,font=\normalsize]

\tikzset{every loop/.style={min distance=5mm,in=140,out=113,looseness=2}}
\begin{tikzpicture}[>=latex, thick,scale=1.1, every node/.style={scale=1}]

\node[] at (0, 0.85) [init,label={[label distance=-0.5mm]180:\textcolor{blue}{$\ell_0$}}]         (0) {};
 \node[] at (0, -0.85) [init,node,label={[label distance=-0.5mm]180:\textcolor{blue}{$\ell_1$}}]   (1) {};

 \node[] at (2.75, 0) [node,label=below:\textcolor{blue}{$\ell_2$}]        (2) {};
 \node[] at (4.2, 0) [final,label=below:\textcolor{blue}{$\ell_3$}]    (3) {};

\draw[post] (0) to[]
    node[sloped, above, align= center,xshift=0cm]
    {\small $r_2 \colon \gamma_1 \mapsto \cpp{x}\quad$} (2);
\draw[post] (1) to[] node[sloped, above, align= center,yshift=0cm]
    {\small $ r_1 \colon \top \mapsto \cpp{x}$~ ~}(2);
\draw[post] (2)to[]
    node[align=center,anchor=north, midway]
    {\small $r_3 \colon \gamma_2$} (3);
\draw[rule] (0) to[out=55,in=100,looseness=8]
	    node[align=center,anchor=south, yshift=0cm,xshift=-0.1cm]{\small $sl_1\colon \top$} (0);
\draw[rule] (2) to[out=55,in=100,looseness=8]
	    node[align=center,anchor=south, yshift=0cm,xshift=-0.1cm]{\small $sl_2{:} \top$} (2);
\draw[rule] (3) to[out=55,in=100,looseness=8]
	    node[align=center,anchor=south,yshift=0cm,xshift=-0.1cm]{\small $sl_3{:}\top$} (3); 

\end{tikzpicture}\caption{Threshold automaton modeling the body of the loop in the 
      protocol from Fig.~\ref{fig:st}. Symbols $\gamma_1, \gamma_2$ stand for the
      threshold guards $x \ge (t+1) - f$ and $x \ge (n-t) - f$, 
      where $n$ and $t$ are as in Fig.~\ref{fig:st}, and $f$ is the actual
     number of faulty processes. The shared variable $x$ models the number of ECHO messages sent by 
     correct processes. Processes with $\mathit{myval}_i=b$ (line~1) start in location 
     $\ell_b$ (in green). Rules $r_1$ and $r_2$ model sending ECHO at lines 7
     and 12. The self-loop rules $sl_1,\ldots, sl_3$ are stuttering steps.
     }\label{fig:stexample}
  \end{minipage}
\end{figure}

We introduce threshold automata, illustrating the definitions on the example of Figure~\ref{fig:stexample}, a model of the Byzantine agreement protocol of Figure \ref{fig:st}. 

\subsubsection*{Environments.} Threshold automata are defined relative to an \emph{environment} $\Env = (\paraset, \ResCond, \syssize)$, where $\paraset$ is a set of \emph{parameters} ranging over $\NatZero$, $\ResCond \subseteq \NatZero^{\paraset}$ is a \emph{resilience condition} expressible as an integer linear formula, and $\syssize \colon \ResCond \rightarrow \NatZero$ is a linear function. Intuitively, a valuation of $\paraset$ determines the number of processes of different kinds (e.g., faulty) executing the protocol, and $\ResCond$ describes the admissible combinations of parameter values. Finally, $\syssize$ associates to a each admisible combination, the number of copies of the automaton that are going to run in parallel, or, equivalently, the number of processes explicitly modeled. In a Byzantine setting, faulty processes behave arbitrarily, and so we do not model them explicitly; in this case, the system consists of one copy of the automaton for every correct process. In the crash fault model, processes behave correctly until they crash, they must be modeled explicitly, and the system has a copy of the automaton for each process, faulty or not. 

\begin{example}
In the threshold automaton of Figure~\ref{fig:stexample}, the parameters are $n$, $f$, and $t$, describing the number of processes, the number of faulty processes, and the maximum possible number of faulty processes, respectively. The resilience condition is the set of triples $(i_n, i_f, i_t)$ such that $i_n/3> i_t \geq i_f$; abusing language, we identify it with the constraint $n/3 > t \geq f$.  The function $\syssize$ is given by $\syssize(n,t,f)= n-f$, which is the number of correct processes. 
\end{example}

\subsubsection*{Threshold automata.} 
A \emph{threshold automaton} over an environment $\Env$ is a tuple $\TA=(\local, \initlocal, \globset, \ruleset)$, where
$\local$ is a nonempty, finite set of  \emph{local states} (or \emph{locations}), $\initlocal\subseteq\local$ is a nonempty subset  of \emph{initial locations}, $\globset$ is a set of \emph{global variables} ranging over $\NatZero$, and $\ruleset$ is a set of \emph{transition rules} (or just \emph{rules}), formally described below.

A \emph{transition rule} (or just a \emph{rule}) is a tuple $r = (\fromstate, \tostate,  \precond, \update)$, where $\fromstate$ and $\tostate$ are the \emph{source} and \emph{target} locations, $\precond \colon \paraset \cup \globset \rightarrow \{\mathit{true}, \mathit{false}\}$ 
is a conjunction of \emph{threshold guards}, and $\update \colon \globset \rightarrow \{0,1\}$ is an \emph{update}. We often let $r.\fromstate, r.\tostate, r.\precond, r.\update$ denote the components of $r$.
Intuitively, $r$ states that a process can move from $\fromstate$ to $\tostate$ if the current values of $\paraset$ and $\globset$ satisfy $\varphi$, and when it moves it updates the current valuation $\vars$ of $\globset$ by performing the update $\vars := \vars + \update$. Since all components of $\update$ are nonnegative, the values of global variables never decrease.  A \emph{threshold guard} $\varphi$ has one of the following two forms:
\begin{itemize}
\item $x\ge a_0 + a_1 \cdot p_1 + \ldots + a_{|\paraset|} \cdot p_{|\paraset|}$, called a \emph{rise guard}, or
  \item $x < a_0 + a_1 \cdot p_1 + \ldots + a_{|\paraset|} \cdot p_{|\paraset|}$, called a \emph{fall guard},
\end{itemize}
where~$x$ $\in \globset$ is a shared variable, 
$p_1,\ldots, p_{|\paraset|}\in \paraset$ are the parameters, and
$a_0,a_1,\ldots,a_{|\paraset|}\in \Rational$ are rational coefficients. 
Since global variables are initialized to $0$, and they never decrease, once a rise (fall) guard becomes true (false) it stays true (false). We call this property \emph{monotonicity of guards}.
We let $\PrecondU$, $\PrecondL$, and $\Precond$ denote the sets of rise guards, fall guards, and all guards of
$\TA$.  

\begin{example}
The rule $r_2$ of Figure~\ref{fig:stexample}  has $\ell_0$ and $\ell_2$ as source and target
locations, $x \geq (t+1)-f$ as guard, and the number $1$ as update (there is only one
shared variable, which is increased by one). 
\end{example}

\subsubsection*{Configurations and transition relation.}
A \emph{configuration} of $\TA$ is a triple $\sigma=(\counters,\vars,\param)$ where 
$\counters  \colon \local \rightarrow \NatZero$ describes the number of processes at each location, and $\vars \in \NatZero^\numglob$ and $\param  \in \ResCond$ are valuations of the global variables and the parameters. In particular, $\sum_{\ell \in \local} \counters(\ell)= \syssize(\param)$ always holds.  A configuration is \emph{initial} if $\counters(\ell) =0$ for every $\ell \notin \initlocal$, and $\vars = \vec{0}$. We often let $\sigma.\counters, \sigma.\vars, \sigma.\param$ denote the components of $\sigma$.

A configuration $\sigma=(\counters,\vars,\param)$  
\emph{enables} a rule $r = (\fromstate, \tostate,  \precond, \update)$ if $\counters(\fromstate) > 0$, and $(\vars, \param)$ satisfies the guard $\precond$, i.e., substituting $\vars(x)$ for $x$ and $\param(p_i)$ for $p_i$ in $\precond$ yields a true expression, denoted by $\sigma\models\varphi$. If $\sigma$ enables $r$, then $\TA$ can \emph{move} from $\sigma$ to the configuration $r(\sigma)=(\counters',\vars',\param')$ defined as follows:
(i)  $\param' = \param$, (ii) $\vars' = \vars + \update$, and (iii) $\counters = \counters + \vec{v}_r$, where $\vec{v}_r (\fromstate) =-1$, $\vec{v}_r (\tostate) =+1$, and $\vec{v}_r = 0$ otherwise.
We let $\sigma \rightarrow r(\sigma)$ denote that $\TA$ can move from $\sigma$ to $r(\sigma)$.

\subsubsection*{Schedules and paths.}
A \emph{schedule} is a (finite or infinite) sequence of rules. 
A schedule $\tau = r_1, \dots, r_m$ is \emph{applicable} to configuration $\sigma_0$ if there is a sequence of configurations $\sigma_1,\dots, \sigma_m$ such that $\sigma_i =
     r_{i} (\sigma_{i-1})$ for $1 \le i \le m$, and we define $\tau(\sigma_0):= \sigma_m$.
We let $\sigma \xrightarrow{*} \sigma'$ denote that $\tau(\sigma) = \sigma'$ for some schedule $\tau$, and say that $\sigma'$ is \emph{reachable} from $\sigma$. Further we let $\tau \cdot \tau'$ denote the concatenation
     of two schedules $\tau$ and~$\tau'$, and, given $\mu \geq 0$, let $\mu \cdot \tau$ the concatenation of $\tau$ with itself $\mu$ times.

A \emph{path} or \emph{run} is a finite or infinite sequence $\gst_0, r_1, \gst_1, \dots, \gst_{k-1},
     r_k, \sigma_k, \dots$ of alternating configurations and rules such that
     $\gst_i = r_i(\gst_{i-1})$ for every $r_i$ in 
     the sequence. If $\tau=r_1,\ldots,r_{|\tau|}$ is applicable to $\sigma_0$, then we let $\finpath{\gst_0}{\tau}$
     denote the path $\gst_0, r_1, \gst_1,
     \dots, r_{|\tau|}, \gst_{|\tau|}$ with 
     $\gst_i = r_i(\gst_{i-1})$, for $1 \le i \le |\tau|$.
Similarly, if $\tau$ is an infinite schedule.
Given a path $\infpath{\gst}{\tau}$, the set of all configurations
     in the path is denoted by~$\setconf\sigma\tau$.

\section{Coverability and Parameterized Coverability} \label{sec:coverability}
We say that configuration~$\sigma$ \emph{covers} location~$\ell$ if $\sigma.\counters(\ell)>0$. 
We consider the following two \emph{coverability} questions in threshold automata:	
	
\begin{definition}[(Parameterized) coverability]
The \emph{coverability problem} consists of deciding, given a threshold automaton~$\TA$, a 
location $\ell$ and an initial configuration $\gst_0$, if some configuration reachable
from $\gst_0$ covers $\ell$. The \emph{parameterized coverability problem} consists of deciding, given 
$\TA$ and $\ell$, if there is an initial configuration $\gst_0$ and a configuration reachable
from $\gst_0$ that covers $\ell$.
\end{definition}

Sometimes we also speak of the \emph{non-parameterized coverability} problem, instead of the coverability problem, to avoid confusion. 
We show that both problems are NP-hard,
	even when the underlying threshold automaton is acyclic.
	In the next section, we show that the reachability and parameterized reachability problems (which subsume the 
	corrersponding coverability problems) are both in NP.

\begin{theorem} \label{theorem:paramcov}
Parameterized coverability in threshold automata is NP-hard, even 
for acyclic threshold automata with only constant guards (i.e., guards of the form $x \geq a_0$ and $x < a_0$).
\end{theorem}

\begin{proof}(Sketch.)
	We prove NP-hardness of parameterized coverability
	by a reduction from 3-SAT. The reduction is as follows:
	(See Figure~\ref{fig:paramcov} for an
	illustrative example).
	Let $\varphi$ be a 3-CNF formula with variables $x_1,\dots,x_n$.
	For every variable $x_i$ we will have two shared variables
	$y_i$ and $\bar{y_i}$. For every clause $C_j$, we will have a shared 
	variable $c_j$.
	Intuitively, each process begins at some state $\ell_i$ and then
	moves to either $\top_i$ or $\bot_i$ by firing either
	$(\ell_i,\top_i,\bar{y_i}<1,\cpp{y_i})$ or 
	$(\ell_i,\bot_i,y_i<1,\cpp{\bar{y_i}})$ respectively. 
	Moving to $\top_i$ ($\bot_i$ resp.) means that the process has guessed the value 
	of the variable $x_i$ to be true (false resp).
	Once it has chosen a truth value, it then
	increments the variables corresponding to all the clauses which
	it satisfies and moves to a location $\ell_\textit{mid}$. If it happens that all the guesses were correct,
	a final rule gets unlocked and processes can move from $\ell_{\textit{mid}}$ to $\ell_F$. 
	The key property we need to show is that if some process moves to 
	$\top_i$ then no other process can move to $\bot_i$ (and vice versa). 
	This is indeed the case because if a process moves to $\top_i$ from $\ell_i$, 
	it would have fired the 
	rule $(\ell_i,\top_i,\bar{y_i}<1,\cpp{y_i})$ which increments the shared variable~$y_i$, and so falsifies the guard of the corresponding rule $(\ell_i,\bot_i,y_i<1,\cpp{\bar{y_i}})$, and therefore no process can fire it.
	Similarly, if $(\ell_i,\bot_i,y_i<1,\cpp{\bar{y_i}})$ is fired, no process can fire $(\ell_i,\top_i,\bar{y_i}<1,\cpp{y_i})$.
	
\end{proof}

\begin{figure}
	\begin{center}
		\tikzstyle{node}=[circle,draw=black,thick,minimum size=4.3mm,inner sep=0.75mm,font=\normalsize]
		\tikzstyle{edgelabelabove}=[sloped, above, align= center]
		\tikzstyle{edgelabelbelow}=[sloped, below, align= center]
		\begin{tikzpicture}[->,node distance = 1cm,scale=0.8, every node/.style={scale=0.8}]
		\node[node] (dec1)[label=left:\textcolor{blue}{$\ell_1$}] {};
		\node[node] at ($(dec1)+(4,0.4)$) [label=right:\textcolor{blue}{$\top_1$}] (true1){};
		\node[node] at ($(dec1)+(4,-0.4)$) [label=right:\textcolor{blue}{$\bot_1$}](false1) {};
		
		\node[node, below = of dec1] (dec2)[label=left:\textcolor{blue}{$\ell_2$}] {};
		\node[node, below = of true1] (true2)[label=right:\textcolor{blue}{$\top_2$}] {};
		\node[node, below = of false1] (false2)[label=right:\textcolor{blue}{$\bot_2$}] {};
		
		\node[node, below = of dec2] (dec3) [label=left:\textcolor{blue}{$\ell_3$}]{};
		\node[node, below = of true2] (true3)[label=right:\textcolor{blue}{$\top_3$}]{};
		\node[node, below = of false2] (false3)[label=right:\textcolor{blue}{$\bot_3$}] {};
		
		\draw(dec1) edge[edgelabelabove] node{$\bar{y_1} < 1 \mapsto \cpp{y_1}$} (true1);
		\draw(dec1) edge[edgelabelbelow] node{$y_1 < 1 \mapsto \cpp{\bar{y_1}}$} (false1);
		
		\draw(dec2) edge[edgelabelabove] node{$\bar{y_2} < 1 \mapsto \cpp{y_2}$} (true2);
		\draw(dec2) edge[edgelabelbelow] node{$y_2 < 1 \mapsto \cpp{\bar{y_3}}$} (false2);
		
		\draw(dec3) edge[edgelabelabove] node{$\bar{y_3} < 1 \mapsto \cpp{y_3}$} (true3);
		\draw(dec3) edge[edgelabelbelow] node{$y_3 < 1 \mapsto \cpp{\bar{y_3}}$} (false3);
		
		\node[node, right = 6.5cm of dec2] (mid)[label=above:\textcolor{blue}{$\ell_{\textit{mid}}$}] {};
		
		\draw(true1) edge[edgelabelabove] node{$\cpp{c_1}$} (mid);
		\draw(false1) edge[edgelabelbelow] node{$\cpp{c_2}$} (mid);
		\draw(false2) edge[edgelabelabove] node{$\cpp{c_1}\wedge \cpp{c_2}$} (mid);
		\draw(true3) edge[edgelabelabove] node{$\cpp{c_1}$} (mid);
		\draw(false3) edge[edgelabelbelow] node{$\cpp{c_2}$} (mid);
		
		\node[node, right =2.5cm of mid] (F) [label=above:\textcolor{blue}{$\ell_F$}]{};
		\draw(mid) edge[edgelabelabove] node{$c_1 \ge 1 \wedge c_2\ge 1$} (F);
		
		\end{tikzpicture}
		\caption{Threshold automaton~$\TA_\varphi$ corresponding to the formula $\varphi = (x_1 \vee \neg x_2 \vee x_3) \wedge (\neg x_1 \vee \neg x_2 \vee \neg x_3)$. 
			Note that setting $x_1$ to true and $x_2$ and $x_3$ to false 
			satisfies~$\varphi$. 
			Let~$\sigma_0$ be the initial configuration obtained by having~1 process in each initial location $\ell_i$, $1\le i\le 3$, 
			and~0 in every other location.
			From $\ell_1$ we increment $y_1$ and from $\ell_2$ and $\ell_3$
			we increment $\bar{y_2}$ and $\bar{y_3}$ respectively, thereby
			making the processes go to $\top_1,\bot_2,\bot_3$ respectively. From there
			we can move all the processes to $\ell_\textit{mid}$, at which
			point the last transition gets unlocked and we can cover $\ell_F$.}
		\label{fig:paramcov}
	\end{center}
\end{figure}

A modification of  the same construction proves

\begin{theorem} \label{theorem:nonparamcov}
	The coverability problem is NP-hard even for acyclic threshold
	automata with only constant rise guards (i.e., guards of the form $x \geq a_0$).
\end{theorem}
		
\paragraph*{Constant rise guards.}

Theorem \ref{theorem:nonparamcov} puts strong constraints to the class of $\TA$s for which 
parameterized coverability can be polynomial, assuming P $\neq$ NP. We identify an interesting polynomial case.

\begin{definition}
	\label{def:mult}
	An environment $\Env = (\paraset, \ResCond, \syssize)$ is \emph{multiplicative} for a $\TA$ if for every $\mu \in \mathbb{N}_{> 0}$ (i) for every valuation $\param \in \ResCond$ we have $\mu \cdot \param \in \ResCond$ and $\syssize(\mu \cdot \param) = \mu \cdot \syssize(\param)$, 
	and (ii) for every guard $\varphi := x \ \Box \ a_0 + a_1 p_1 + a_2 p_2 + \dots a_k p_k$ in $\TA$ (where $\Box \in \{\ge, <\}$), if $(y,q_1,q_2,\dots,q_k)$
	is a (rational) solution to $\varphi$ then $(\mu \cdot y, \mu \cdot q_1, \dots, \mu \cdot q_k)$ is also a solution to $\varphi$.
\end{definition}

Multiplicativity is a very mild condition. To the best of our knowledge, all algorithms discussed in the literature,
and all benchmarks of \cite{ELTLFT}, have multiplicative environments.  
For example, in Figure \ref{fig:stexample}, if the resilience condition $t < n/3$ holds for a pair $(n, t)$, then
it also holds for $(\mu \cdot n, \mu \cdot t)$; similarly, the function $\syssize(n,t,f)=n-f$ also satisfies
$\syssize(\mu \cdot n, \mu \cdot t, \mu \cdot f) = \mu \cdot n -\mu \cdot f = \mu \cdot \syssize(n,t,f)$. 
Moreover, if $x \ge t+1-f$ holds in $\sigma$, then we also have $\mu\cdot x \ge \mu\cdot t +1 -\mu\cdot f$ in $\mu \cdot \sigma$.
Similarly for the other guard $x\ge n-t-f$.

This property allows us to reason about multiplied paths in large systems.
Namely, condition (ii) from Definition~\ref{def:mult} yields that if a rule is enabled in $\sigma$, it is also enabled in $\mu\cdot\sigma$.
This plays a crucial role in Section~\ref{sec:liveness} where we need the fact that a counterexample in a small system implies a counterexample in a large system.



\begin{theorem} \label{theorem:cons-rise-complete}
Parameterized coverability of threshold automata over multiplicative environments with only constant rise guards is P-complete.
\end{theorem}
\begin{proof}(Sketch.)  P-hardness is proved by giving a logspace-reduction from the Circuit Value problem (\cite{CVP}) which is well known to be P-hard. In the following, we sketch the proof of inclusion in P.

Let $\TA = (\local, \initlocal,\globset, \ruleset)$ be a threshold automaton over a multiplicative
environment $\Env = (\paraset, \ResCond, \syssize)$ such that the guard of each transition in $\ruleset$ is a constant rise guard.
We construct the set $\widehat{\local}$ of locations that can be reached by at least one process,  and the set 
of transitions $\widehat{\ruleset}$ that can occur, from at least one initial configuration. 
We initialize two variables $X_\local$ and $X_\ruleset$ by $X_\local := \initlocal$ and $X_\ruleset := \emptyset$, and repeatedly update 
them until a fixed point is reached, as follows:
\begin{itemize}
\item If there exists a rule $r = (\ell,\ell',\true,\vec{u}) \in \ruleset \setminus X_\ruleset$ such that $\ell \in X_\local$, then set $X_\local := X_\local \cup \{\ell'\}$ and $X_\ruleset := X_\ruleset \cup \{r\}$.
\item If there exists a rule $r=(\ell,\ell', (\land_{1 \le i \le q} \ x_i \geq c_i), \vec{u}) \in \ruleset \setminus X_\ruleset$ such that $\ell \in X_\local$,  and there exists rules $r_1,r_2,\dots,r_q$ such that each
	$r_i = (\ell_i,\ell'_i,\varphi_i,\vec{u}_i) \in X_\ruleset$
	and $\vec{u}_i[x_i] > 0$, then set $X_\local := X_\local \cup \{\ell'\}$ and $X_	\ruleset := X_\ruleset \cup \{r\}$.
\end{itemize}

We prove in the Appendix that after termination $X_\local = \widehat{\local}$ holds. Intuitively, multiplicativity ensures that if a reachable configuration enables a rule, there are reachable configurations from which the rule can occur arbitrarily  many times. This shows that any path of rules constructed by the algorithm is executable.
\end{proof}


\section{Reachability} \label{sec:reachability}

We now consider reachability problems for threshold automata. 
Formally, we consider the following two versions of the reachability problem:
	
\begin{definition}[(Parameterized) reachability]
The \emph{reachability problem} consists of deciding,  given a threshold automaton~$\TA$,  two sets $\local_{=0},  \local_{>0}$ of locations, and an initial configuration $\gst_0$, if some configuration $\gst$ reachable from  $\gst_0$ satisfies $\gst.\counters(\ell) = 0$ for every 
$\ell \in \local_{=0}$ and $\gst.\counters(\ell) > 0$ for every $\ell \in \local_{>0}$.
The \emph{parameterized reachability problem} consists of deciding, given $\TA$ and $\local_{=0},  \local_{>0}$, if there is an initial configuration $\gst_0$ such that some $\gst$ reachable from  $\gst_0$ satisfies $\gst.\counters(\ell) = 0$ for every 
$\ell \in \local_{=0}$ and $\gst.\counters(\ell) > 0$ for every $\ell \in \local_{>0}$.
\end{definition}

Notice that the reachability problem clearly subsumes the coverability problem and hence, in the sequel, we will only be concerned with proving that both problems are in NP. This will be a consequence of our main result, showing that the reachability relation of threshold automata can be characterized as an existential formula of Presburger arithmetic. This result has several other consequences. In Section \ref{sec:liveness} we use it to give a new model checking algorithm for 
the fault-tolerant logic of~\cite{ELTLFT}. In Section \ref{sec:experiments} we report on an implementation whose runtime compares favorably with previous tools. 

\subsubsection*{Reachability relation as an existential Presburger formula.}
Fix a threshold automaton $\TA=(\local, \initlocal, \globset, \ruleset)$ over an environment $\Env$. We construct an existential Presburger arithmetic formula $\phi_\mathit{reach}$ with $(2\numlocal+2\numglob+2|\paraset|)$ free variables such that $\phi_\mathit{reach}(\sigma,\sigma')$ is true if{}f $\sigma'$ is reachable from $\sigma$. 


Let the \emph{context} of a configuration~$\sigma$, denoted by $\statectx(\sigma)$,
be the set of all \emph{rise} guards that evaluate to true and all \emph{fall} guards that 
evaluate to false in~$\sigma$. Given a schedule $\tau$, we say that the path
$\infpath{\sigma}{\tau}$ is \emph{steady} if all the configurations it visits have the same context.
By the monotonicity of the guards of threshold automata, $\infpath{\sigma}{\tau}$ is steady if{}f its endpoints have the same context, i.e., if{}f $\statectx(\sigma)=\statectx(\tau(\sigma))$. We have the following proposition:
\begin{proposition}\label{prop:gen-to-steady}
Every path of a threshold automaton with $k$ guards is the 
concatenation of at most $k {+} 1$ steady paths.
\end{proposition}
\noindent Using this proposition, we first construct a formula $\phi_\mathit{steady}$ such that $\phi_\mathit{steady}(\sigma,\sigma')$ holds if{}f there is a \emph{steady} path $\infpath{\sigma}{\tau}$ such that $\tau(\sigma)=\sigma'$. 

\subsubsection*{The formula $\phi_\mathit{steady}$.} For every rule $r\in\ruleset$,  let $x_r$ be a variable ranging over non-negative integers. Intuitively, the value of~$x_r$ will represent the number of times~$r$ is fired during the (supposed) path from~$\sigma$ to~$\sigma'$. Let $X=\{x_r\}_{r\in\ruleset}$. We construct $\phi_\mathit{steady}$ step by step, specifying necessary conditions for $\sigma,\sigma'$ and $X$ to satisfy the existence of the steady path, which in particular implies that $\sigma'$ is reachable from $\sigma$.

\paragraph{Step 1.} $\sigma$ and $\sigma'$ must have the same values of the parameters, which must satisfy the resilience condition, the same number of processes, and  the same context: 
$$\phi_\mathit{base}(\sigma,\sigma') \equiv \;  \sigma.\param = \sigma'.\param \; \wedge \;
\ResCond(\sigma.\param) \; \wedge \; \syssize(\sigma.\param) = 
\syssize(\sigma'.\param) \; \wedge \;
\statectx(\sigma)=\statectx(\sigma').$$

\paragraph{Step 2.} For a location $\ell\in\local$, 
let $out^\ell_1,\ldots,out^\ell_{a_\ell}$ be all outgoing rules from~$\ell$ and let $in^\ell_1,\dots,in^\ell_{b_\ell}$ be all incoming rules to~$\ell$. The number of processes in $\ell$ after the execution of the path is the initial number, plus the incoming processes, minus the outgoing processes. Since $x_r$ models the number of times the rule $r$ is fired, we have
$$	\phi_\local(\sigma,\sigma',X) \equiv \quad
	\bigwedge_{\ell \in \local} \ \left(\sum_{i=1}^{a_\ell} x_{in^\ell_i} - \sum_{j=1}^{b_\ell} x_{out^\ell_j} =
	\sigma'.\counters(\ell) - \sigma.\counters(\ell)\right)
$$
\paragraph{Step 3.} Similarly, for the shared variables we must have:
$$  \phi_\globset(\sigma,\sigma',X)  \equiv \quad
	\bigwedge_{z\in\globset} \ \left(\sum_{r\in\ruleset} (x_{r} \,\cdot\, r.\update[z]) = \sigma'.\vars[z] - \sigma.\vars[z]\right)
$$
\paragraph{Step 4.} Since $\infpath\sigma\tau$ must be steady, if a rule is fired along $\infpath\sigma\tau$ then its guard must be true in~$\sigma$ and so
$$\phi_\ruleset(\sigma,X)  \equiv \quad \bigwedge_{r \in \ruleset} \ x_r > 0 \ \Rightarrow \ (\sigma\models r.\varphi)$$
\paragraph{Step 5.} Finally, for every rule~$r$ that occurs in $\infpath\sigma\tau$, the path 
must contain a ``fireable'' chain leading to $r$, i.e., a set of rules $S=\{r_1,\dots,r_s\} \subseteq \ruleset$ such that  all rules of $S$ are executed in $\infpath\sigma\tau$, there is a process in $\sigma$ at $r_1.\fromstate$,  and the rules $r_1,\dots,r_s$  form a chain leading from 
$r_1.\fromstate$ to $r.\fromstate$. We capture this by the constraint
$$  \phi_{\textit{appl}}(\sigma,X) \equiv \quad
	\bigwedge_{r \in \ruleset} \; \left( x_r > 0 \ \Rightarrow \ \bigvee_{S=\{r_1,r_2,\dots,r_s\} \subseteq \ruleset} \phi_\mathit{chain}^r(S,\sigma,X)\right)
$$
where
$$
\phi_\mathit{chain}^r(S,\sigma,X) \equiv \;
\bigwedge_{r \in S} x_{r} > 0 \;\land \;  \sigma.\counters(r_1.\fromstate) > 0 \; \land \;
 \bigwedge_{1 < i \le s} r_{i-1}.\tostate = r_i.\fromstate \; \land \; r_s = r
$$

\paragraph{Combining the steps. } Define $\phi_\mathit{steady}(\sigma,\sigma')$ as follows:
\begin{eqnarray*}\label{eq:kirchoff}
\phi_\mathit{steady}(\sigma,\sigma') 
& \equiv &  \phi_\mathit{base}(\sigma,\sigma') \ \land \  \\
&    &	\exists X \ge 0.\; 
	\phi_\local(\sigma,\sigma',X)  \land 
	\phi_\globset(\sigma,\sigma',X)  \land 
	\phi_\ruleset(\sigma,X)  \land  
	\phi_{\textit{appl}}(\sigma,X) \ .
\end{eqnarray*}
\noindent where $\exists X \ge 0$ abbreviates $\exists x_{r_1} \ge 0 ,\ldots, \exists x_{r_{|\ruleset|}} \ge 0$. By our discussion, it is clear that if there is a steady path leading from
$\sigma$ to $\sigma'$, then $\phi_\mathit{steady}(\sigma,\sigma')$ is satisfiable. The following theorem proves the converse.

\begin{theorem}\label{th:kirchoffpath}
Let~$\TA$ be a threshold automaton and let $\sigma, \sigma' \in \configs$ be two configurations.
Formula $\phi_\mathit{steady}(\sigma,\sigma')$ is satisfiable if and only if there is a steady schedule~$\tau$ with $\tau(\sigma)=\sigma'$.
\end{theorem}

Observe that, while  $\phi_\mathit{steady}$ has exponential length in $\TA$ when constructed na\"{i}vely (because of the exponentially many disjunctions in $\phi_\textit{appl}$), its satisfiability is in NP. Indeed, we first non-deterministically guess one of the disjunctions for each conjunction of $\phi_\textit{appl}$ and then check in nondeterministic polynomial time that the  (polynomial sized) formula with only these disjuncts is satisfiable. This is possible because existential Presburger arithmetic is known to be in NP~\cite{ExistPres}.

\subsubsection*{The formula $\phi_\mathit{reach}$.} By Proposition~\ref{prop:gen-to-steady}, 
every path from $\sigma$ to $\sigma'$ in a threshold automaton with a set $\Precond$ of guards can be written in the form 
\begin{equation*}
		\sigma = \sigma_0 \xrightarrow{*} \sigma'_0 \rightarrow \sigma_1 \xrightarrow{*} \sigma'_1 \rightarrow \sigma_2 \dots \sigma_{K} \xrightarrow{*} \sigma'_{K} = \sigma'
	\end{equation*}
where $K=|\Precond| + 1$, and $\sigma_i \xrightarrow{*} \sigma'_i$ is a steady path for each $0\le i \le K$.
It is easy to see from the definition of the 
	transition relation between configurations that
	we can construct a polynomial sized existential 
	Presburger formula $\phi_\mathit{step}$ such that 
	$\phi_\mathit{step}(\sigma,\sigma')$
	is true iff $\sigma'$ can be reached from $\sigma$ by firing
	at most one rule.
Thus, we  define $\phi_\mathit{reach}(\sigma,\sigma')$ to be 
	\begin{equation*}
		\exists \sigma_0, \sigma'_0, \dots, \sigma_{K}, \sigma'_{K} 
		\left(\sigma_0 = \sigma \land \sigma'_{K} = \sigma' \land \bigwedge_{0 \le i \le K} \phi_\mathit{steady}(\sigma_i,\sigma'_i) 
		\land  \bigwedge_{0 \le i \le K-1}
		\phi_\mathit{step}(\sigma'_i,\sigma_{i+1}) \right)
	\end{equation*}

\begin{theorem} \label{th:main}
Given a threshold automaton~$\TA$, there is an existential Presburger formula $\phi_\mathit{reach}$ such that $\phi_\mathit{reach}(\sigma,\sigma')$ holds if{}f 
$\sigma \xrightarrow{*} \sigma'$. 
\end{theorem}

As deciding the truth of  existential Presburger formulas is in NP, we obtain:
\begin{corollary} \label{cor:reach}
The reachability and parameterized reachability problems are in~NP.
\end{corollary}

\begin{remark}
In~\cite{Reach} an algorithm was given for parameterized reachability of threshold automata in
which the updates of all rules contained in loops are equal to $\vec{0}$. 
Our algorithm does not need this restriction.
\end{remark}

\section{Safety and Liveness}\label{sec:liveness}

We recall the definition of \emph{Fault-Tolerant
Temporal Logic} ($\ELTLFT$), the fragment of LTL used in~\cite{ELTLFT} to specify and verify properties of a large number of fault-tolerant algorithms. $\ELTLFT$ has the following syntax, where $\mathit{S} \subseteq \local$ is a set of locations and $\mathit{guard} \in \Precond$ is a guard:  
\begin{align*}
    \psi &::=
    \mathit{pf} \ |\ \ltlG \psi\ |\ \ltlF \psi \ |\  \psi \wedge \psi
    &
    \mathit{cf} &::= \mathit{S}=0 \mid  \neg (\mathit{S}=0) \mid \mathit{cf} \wedge \mathit{cf} 
    \\
    \mathit{pf} &::= \mathit{cf} \ |\ \mathit{gf} \Rightarrow \mathit{cf}
    &
    \mathit{gf} &::= \mathit{guard}
    \ |\ \mathit{gf}\wedge \mathit{gf} \ |\ \mathit{gf} \vee \mathit{gf}
\end{align*}
An infinite path $\infpath{\sigma}{\tau}$ starting at $\sigma=(\counters,\vars,\param)$,
 satisfies $\mathit{S}=0$ if $\counters(\ell) = 0$ for every $\ell \in S$, and $\mathit{guard}$ if $(\vars, \param)$ satisfies $\mathit{guard}$. The rest of the semantics is standard.  The negations of specifications of the benchmarks  \casest{} can be expressed in $\ELTLFT$, as we are interested in finding 
 possible violations.

\begin{example}
One specification of the algorithm from Figure~\ref{fig:st} is that 
	if $\mathit{myval}_i=1$ for every correct process $i$, then
	eventually $\mathit{accept}_j = \mathit{true}$ for some correct 
	process $j$.
In the words of the automaton from Figure~\ref{fig:stexample}, 
	a violation of this property would mean that 
	initially all correct processes are in location~$\ell_1$,
	but no correct process ever reaches location~$\ell_3$.
In $\ELTLFT$ we write this as 
\[\{\ell_0, \ell_2, \ell_3\} =0 \;\wedge\; \ltlG (\{\ell_3\} =0) \, .\]
This has to hold under the fairness constraint
	$$
	\ltlG \ltlF \bigg( \, ((x \ge t+1 \vee  x\ge n-t) \Rightarrow \{\ell_0\}\!=\!0)  \; \wedge \;  \{\ell_1\}\!=\!0 \; \wedge \;  (x \ge n-t \Rightarrow \{\ell_2\}\!=\!0) \, \bigg).
	$$
	As we have self-loops at locations~$\ell_0$ and $\ell_2$, a process could stay forever in one of these two states, even if it has collected enough messages, i.e., if $x\ge t+1$ or $x\ge n-t$.
	This is the behavior that we want to prevent with such a fairness constraint.
	Enough sent messages should force each process to progress, so the location eventually becomes empty. 
	Similarly, as the rule leading from $\ell_1$ has a trivial guard, we want to make sure that all processes starting in $\ell_1$ eventually (send a message and) leave $\ell_1$ empty, as required by the algorithm.
\end{example}

In this section we study the following problem:
\begin{definition}[Parameterized safety and liveness]
	Given a threshold automaton~$\TA$ and a formula~$\varphi$ in
	$\ELTLFT$,
	check whether there is an initial configuration $\sigma_0$ and an infinite schedule~$\tau$ applicable to~$\sigma_0$ such
	that~$\infpath{\sigma_0}{\tau} \models \varphi$.
\end{definition}

Since parameterized coverability is NP-hard, it follows
that parameterized safety and liveness is also NP-hard. We prove that for automata with \emph{multiplicative environments} 
(see Definition \ref{def:mult}) 
parameterized safety and liveness is~in~NP.

\begin{theorem}~\label{th:liveness}
Parameterized safety and liveness of threshold automata with multiplicative environments is in NP.
\end{theorem}

\begin{figure}[t]
\centering
\begin{tikzpicture}[font=\small,>=latex];
    \tikzstyle{node}=[circle,fill=black,minimum size=1mm,inner sep=0cm];
    \tikzstyle{cut}=[cross out,thick,draw=red!90!black,
    minimum size=0.15cm,inner sep=0mm,outer sep=.1mm];
    \tikzstyle{path}=[-];
    \tikzstyle{Gfin}=[-, very thick, blue];
    
    \begin{scope}[scale=0.7]
    \node[node,label={below:$a$}] (init) at (0, 0) {};
    \node[node,label={[xshift=0mm]below:$\ltlF b$}] (01) at (1.5,.75) {};
    \node[node,label={[xshift=0mm]above:$\ltlF c$}] (02) at (1.5,-.75) {};
    \node[node,label={below:$\mathit{loop_{st}}$}] (ls) at (3,0) {};
    \node[node,label={[xshift=0mm]below:$\ltlG\ltlF e$}] (031) at (4.5,.75) {};
    \node[node,label={below:$\mathit{loop_{end}}$}] (le) at (6, 0) {};


    \draw[->] (ls) edge (le);
    \draw[->] (init) edge (01);
    \draw[->] (init) edge (02);

    \draw[->] (ls) -- (031);
    \draw[->] (031) -- (le);

    \draw[->] (init) -- (ls);
    \draw[->] (01) -- (ls);
    \draw[->] (02) -- (ls);
    \end{scope}
    
     \begin{scope}[xshift=6cm, yshift=0cm]
     
    \node[node] (0) at (0, 0) {};
    
    \foreach \x/\n in {.6/A, 1.2/B, 1.8/C, 2.4/D, 3.0/E, 3.6/F}
    \node[cut] (\n) at (\x,0) {}; 
    
    \draw[path] (0) -- (A);
    \draw[path] (A) -- (B);
    \draw[path] (B) -- (C);
    \draw[path] (C) -- (D);
    \draw[Gfin] (D) -- (E);
    \draw[Gfin] (E) -- (F);
    \draw[->] (F) edge[bend right=55,looseness=1.7] (D);
    
    \draw[|<->|] ($(A)+(0,-.5)$)
    --node[midway, fill=white, text=black]
    {$d$} ($(F)+(0,-.5)$);
    \node at ($(A)+(0,-.25)$) {$a$};
    \node at ($(B)+(0,-.25)$) {$b$};
    \node at ($(C)+(0,-.25)$) {$c$};
    \node at ($(E)+(0,-.25)$) {$e$};
    \end{scope}

\end{tikzpicture}
\caption{The cut graph of a formula $\ltlF (a \wedge \ltlF b \wedge 
	\ltlF c \wedge \ltlG d \wedge \ltlG \ltlF e)$ (left) 
	and one lasso shape for a chosen topological 
	ordering $a \le \ltlF b \le \ltlF c \le \mathit{loop_{st}} \le \ltlG \ltlF e \le 
	\mathit{loop_{end}}$ (right).}\label{fig:lasso}
\end{figure}

The proof, which can be found in the Appendix, is very technical, and we only give a rough sketch here. The proof relies on two notions introduced in~\cite{ELTLFT}.   First, it is shown in~\cite{ELTLFT} that every $\ELTLFT$ formula is equivalent to a formula in \emph{normal form} of shape $\phi_0 \land \ltlF \phi_1  \land \dots  \land \ltlF \phi_k  \land  \ltlG \phi_{k+1}$, where $\phi_0$ is a propositional formula and $\phi_1,\dots,\phi_{k+1}$ are themselves in normal form. Further, formulas can be put in normal form in polynomial time. The second notion introduced in~\cite{ELTLFT} is the \emph{cut graph} $Gr(\varphi)$ of a formula in normal form. For our sketch it suffices to know that $Gr(\varphi)$ is a directed acyclic graph with two special nodes $\mathit{loop_{st}}$ and $\mathit{loop_{end}}$, and every other node is a subformula of $\varphi$
in normal form (see Figure~\ref{fig:lasso}). 

For a formula $\varphi \equiv \phi_0 \land \ltlF \phi_1 \land \dots \land \ltlF \phi_k \land \ltlG \phi_{k+1}$, we will say that its \emph{local proposition}
is $\phi_0$ and its \emph{global proposition} is the local proposition
of $\phi_{k+1}$.
It is shown in~\cite{ELTLFT} that, given $\varphi=\phi_0 \land \ltlF \phi_1  \land \dots  \land \ltlF \phi_k  \land  \ltlG \phi_{k+1}$, some infinite path satisfies $\varphi$ if{}f there exists a topological ordering 
$v_0,v_1,\dots,v_c = loop_{st},v_{c+1},\dots,v_l = loop_{end}$ of the cut graph and a path
$\sigma_0,\tau_0,\sigma_1,\dots,\sigma_c,\tau_c,\dots,\sigma_{l-1},\tau_{l-1},\sigma_l$ such that, roughly speaking, (among other technical conditions) 
every configuration $\sigma_i$ satisfies the local proposition of $v_i$
and every configuration in $\setconf{\sigma_i}{\tau_i}$ 
satisfies the global proposition of every $v_j$ where $j \le i$.

Using multiplicativity and our main result that reachability is definable
in existential Presburger arithemtic, we show that for every proposition $p$, we can construct an existential formula $\phi_p(\sigma,\sigma')$ such that: If there is a path between $\sigma$ and $\sigma'$, \emph{all} of whose configurations satisfy $p$, then $\phi_p(\sigma,\sigma')$ is satisfiable. Further, if $\phi_p(\sigma,\sigma')$ is satisfiable, then
there is a path between $2 \cdot \sigma$ and $2 \cdot \sigma'$
\emph{all} of whose configurations satisfy $p$. (Here $2 \cdot \sigma = ((2 \cdot \sigma.\counters), (2 \cdot \sigma.\vars), (2 \cdot \sigma.\param))$).
Then, once we have fixed a topological ordering $V = v_0,\dots,v_l$, (among other conditions), we check if there are configurations $\sigma_0,\dots,\sigma_l$
such that for every $i$, $\sigma_i$ satisfies the local proposition of $v_i$ and for every $j \le i$, 
$\phi_{p_j}(\sigma_i,\sigma_{i+1})$ is satisfiable where $p_j$ is the global proposition of $v_j$. Using multiplicativity, we then show that this
procedure is sufficient to check if the given specification $\varphi$ is satisfied.


Our algorithm consists therefore of the following steps: (1) bring $\varphi$ in normal form; (2) construct the cut graph $Gr(\varphi)$; (3) guess a topological ordering of the nodes of $Gr(\varphi)$; (4) for the guessed ordering,
check in nondeterministic polynomial time if the required sequence
$\sigma_0,\dots,\sigma_l$ exists.


\begin{remark}
From an algorithm given in~\cite{ELTLFT} one can infer that
parameterized safety and liveness is in NP for threshold automata with multiplicative environments,
where all cycles are simple, and rules in cycles have update $\vec{0}$.
(The NP bound was not explicitly given in \cite{ELTLFT}.) Our algorithm only requires multiplicativity.
\end{remark}




\section{Synthesis of Threshold Guards} \label{sec:synthesis}

We study the \emph{bounded synthesis} problem for constructing parameterized threshold guards in threshold automata satisfying
a given specification. 
\paragraph{Sketch threshold automata.}
Let an \emph{indeterminate} be a variable that can take values over rational numbers. We consider threshold automata whose guards can contain indeterminates. More precisely, a \emph{sketch} threshold automaton is a tuple $\TA=(\local, \initlocal, \globset, \ruleset)$, just as before, except for the following change.
Recall that in a threshold automaton, a guard is an inequality of one 
of the following two forms:
$$x\ge a_0 + a_1 \cdot p_1 + \ldots + a_{|\paraset|} \cdot p_{|\paraset|}
\;\mbox{ or }\;
x < a_0 + a_1 \cdot p_1 + \ldots + a_{|\paraset|} \cdot p_{|\paraset|}$$
where $a_0, a_1, \dots, a_{|\paraset|}$ are rational numbers.
In a sketch threshold automaton, some of the $a_0, a_1, \dots, a_{|\paraset|}$ can be \emph{indeterminates}. Moreover, indeterminates can be shared between two or more guards.

Given a sketch threshold automaton $\TA$ and an assignment $\mu$ to the indeterminates,
we let $\TA[\mu]$ denote the threshold automaton obtained by substituting the indeterminates by their values in $\mu$. We define the \emph{bounded synthesis} problem: 
\begin{quote}
\noindent Given: An environment $Env$, a sketch threshold automaton $\TA$ with indeterminates $v_1,\dots,v_m$,
a formula $\varphi$ of $\ELTLFT$, and a polynomial $p$. \\
\noindent Decide:  Is there an assignment $\mu$ to $v_1,\dots,v_m$ 
of size $O(p(|\TA|+|\varphi|))$ (i.e., the vector $(\mu(v_1), \ldots, \mu(v_m))$ of rational numbers can be encoded in binary using $O(p(|\TA|+|\varphi|))$ bits) such that $\TA[\mu]$ satisfies $\lnot \varphi$ (i.e., such that for every initial configuration $\sigma_0$ in $\TA[\mu]$, every infinite run starting from $\sigma_0$ satisfies $\lnot \varphi$)?
\end{quote}

We say that an assignment $\mu$ to the indeterminates makes the environment
multiplicative if the conditions of Definition~\ref{def:mult} are satisfied
after plugging in the assignment $\mu$ in the automaton. 
In the following, we will only be concerned with assignments
which make the environment multiplicative.

Since we can guess an assignment in polynomial time, 
by Theorem~\ref{th:liveness} it follows \begin{theorem}\label{th:bound-synthesis-upperbound}
	Bounded synthesis is in $\Sigma_2^p$.
\end{theorem}

\noindent By a reduction from the $\Sigma_2$-SAT problem, we also provide a matching lower~bound.
\begin{theorem} \label{th:bound-synthesis-lowerbound}
	Bounded synthesis is $\Sigma_2^p$-complete.
\end{theorem}


The \emph{synthesis} problem is defined as the bounded synthesis problem, but lifting the
constraint on the size of $\mu$. While we do not know the exact complexity of the synthesis problem,
we can show that, for a large and practically motivated class of threshold automata
introduced in~\cite{LKWB17:opodis}, the synthesis problem reduces to the bounded synthesis problem. We proceed to describe and motivate the class.

The parameter variables of fault-tolerant distributed algorithms usually consist of 
a variable $n$ denoting the number of processes running the algorithm 
and various ``failure'' variables for the number of processes exhibiting
different kinds of failures (for example, a variable $t_1$ might be used to specify the number of 
Byzantine failures, a variable $t_2$ for crash failures, etc.).
The following three observations are made in \cite{LKWB17:opodis}: 
\begin{itemize}
\item[(1)] The resilience condition of these algorithms is usually of the form $n > \sum_{i=1}^k \delta_i t_i$ where $t_i$ are parameter variables and $\delta_i$ are natural numbers. 
\item[(2)] Threshold guards typically serve one of two purposes: to check if at least a certain fraction of the processes sends a message (for example,
$x > n/2$ ensures that a strict majority of processes has sent a message), or to bound the number of processes that crash. 
\item[(3)] The coefficients of the guards are rational numbers with small denominators (typically at most 3).
\end{itemize}
By (2), the structure of the algorithm guarantees that the value of a variable $x$ never goes beyond $n$, the number of processes. Therefore, given a threshold guard template $x \bowtie \vec{u} \cdot \vec{\pi} + v$, where $\vec{u}$ is a vector of indeterminates, $\vec{\pi}$ is a vector of parameter variables, $v$ is an indeterminate, and $\bowtie$ is either $\ge$ or $<$, we are only interested in assignments $\mu$ of $\vec{u}$ and $v$ which satisfy $0 \le \mu(\vec{u}) \cdot \nu(\vec{\pi}) + \mu(v) \le n$ for every valuation $\nu(\vec{\pi})$ of $\vec{\pi}$ respecting the resilience condition. Guards obtained by instantiating guard templates with such a valuation $\mu$ are called \emph{sane guards}~\cite{LKWB17:opodis}.

The following result is proved in \cite{LKWB17:opodis}: Given a resilience condition
$n > \sum_{i=1}^k \delta_i t_i$, and an upper bound $D$ on the denominator of the entries of 
$\mu$ (see (1) and (3) above), the numerators of the entries of $\mu$ are necessarily of polynomial size in $k, \delta_1, \dots, \delta_k$. Therefore, the synthesis problem for sane guards and bounded denominator, as introduced in~\cite{LKWB17:opodis}, reduces to the bounded synthesis problem, and so it can be solved in $\Sigma_2^p$ time. 
Moreover, the reduction used in Theorem \ref{th:bound-synthesis-lowerbound} to prove $\Sigma_2^p$-hardness yields sketch threshold automata with sane guards, and so the the synthesis problem for sane guards and bounded denominator is also $\Sigma_2^p$-complete.

\section{Experimental Evaluation} \label{sec:experiments}
	
Following the techniques presented in this paper, we have verified a number of threshold-based fault-tolerant distributed algorithms.

\paragraph{Benchmarks.}

Consistent broadcast (\textbf{strb})~\cite{ST87:abc} is given in Figure~\ref{fig:st} and its threshold automaton is depicted in Figure~\ref{fig:stexample}.
The algorithm is correct if in any execution either all correct processes or none set accept to true; moreover, if all correct processes start with value~0 then none of them accept, and if all correct processes start with value~1 then they all accept.
The algorithm is designed to tolerate Byzantine failures of less than one third of processes, that is, if $n>3t$.
Folklore Reliable Broadcast (\textbf{frb})~\cite{CT96} that tolerates crash faults and Asynchronous Byzantine agreement (\textbf{aba})~\cite{BrachaT85} satisfy the same specifications as consistent broadcast, under the same resilience condition.

%

{\small
	\setlength{\tabcolsep}{10pt}
\renewcommand{\arraystretch}{1.25}
	\begin{table}[H]
		\centering
	\begin{tabular}{c c | m{0.75cm} | m{0.75cm} | c | c }
		\hline
		\textbf{Input} & \textbf{Case} & \multicolumn{2}{m{1.9cm}|}{{ \textbf{Threshold automaton}}} & \multicolumn{2}{c}{\textbf{Time, seconds}}\\
		\cline{3-6} 
		{} & {(if more than one)} & {\textbf{$|\local|$}} & {\textbf{$|\ruleset|$}} &
		{\textbf{Our tool}} & {\textbf{ByMC}}\\
		\hline
		
		\textbf{nbacg} & {} & {24} & {64} & {11.84} & {10.29}\\
		\textbf{nbacr} & {} & {77} & {1031} & {490.79} & {1081.07}\\
		\textbf{aba} & {Case 1} & {37} & {202} & {251.71} & {751.89}\\
		\textbf{aba} & {Case 2} & {61} & {425} & {2856.63} & {TLE}\\
		\textbf{cbc} & {Case 1} & {164} & {2064} & {MLE} & {MLE}\\
		\textbf{cbc} & {Case 2} & {73} & {470} & {2521.12} & {36.57}\\
		\textbf{cbc} & {Case 3} & {304} & {6928} & {MLE} & {MLE}\\
		\textbf{cbc} & {Case 4} & {161} & {2105} & {MLE} & {MLE}\\
		\textbf{cf1s} & {Case 1} & {41} & {280} & {50.5} & {55.87}\\
		\textbf{cf1s} & {Case 2} & {41} & {280} & {55.88} & {281.69}\\
		\textbf{cf1s} & {Case 3} & {68} & {696} & {266.56} & {7939.07}\\
		\textbf{cf1s} & {hand-coded TA} & {9} & {26} & {7.17} & {2737.53}\\
		\textbf{c1cs} & {Case 1} & {101} & {1285} & {1428.51} & {TLE}\\
		\textbf{c1cs} & {Case 2} & {70} & {650} & {1709.4} & {11169.24}\\
		\textbf{c1cs} & {Case 3} & {101} & {1333} & {TLE} & {MLE}\\
		\textbf{c1cs} & {hand-coded TA} & {9} & {30} & {37.72} & {TLE}\\ 
		\textbf{bosco} & {Case 1} & {28} & {152} & {58.11} & {89.64}\\
		\textbf{bosco} & {Case 2} & {40} & {242} & {157.61} & {942.87}\\
		\textbf{bosco} & {Case 3} & {32} & {188} & {59} & {104.03}\\
		\textbf{bosco} & {hand-coded TA} & {8} & {20} & {20.95} & {510.32}\\
		\hline\\
	\end{tabular}
	\caption{The experiments were run on a machine with Intel\textsuperscript{\textregistered} Core\textsuperscript{\tiny {TM}} i5-7200U CPU with 7.7 GiB memory. 
	The time limit was set to be 5 hours and the memory limit was set to be 7 GiB. TLE (MLE) means that the time limit (memory limit) exceeded
	for the particular benchmark.}	\label{table}	
	\end{table}

}

\vspace{-0.17cm}

Non-blocking atomic commit (\textbf{nbacr})~\cite{Raynal97} and (\textbf{nbacg})~\cite{Gue02} deal with faults using failure detectors. We model this by introducing a special location such that a process is in it if and only if it suspects that there is a failure of the system.

Condition-based consensus (\textbf{cbc})~\cite{MostefaouiMPR03} reaches consensus under the condition that the difference between the numbers of processes initialized with 0 and 1 differ by at least $t$, an upper bound on the number of faults.
We also check algorithms that allow consesus to be achieved in one communication step, such as \textbf{cfcs}~\cite{DobreS06}, \textbf{c1cs}~\cite{BrasileiroGMR01}, as well as Byzantine One Step Consensus \textbf{bosco}~\cite{SongR08}.

\paragraph{Evaluation.}

Table \ref{table} summarizes our results and compares them with the results obtained using the ByMC tool~\cite{KW18}. Due to lack of space,
we have omitted those experiments for which both ByMC and our tool took
less than 10 seconds.

We implemented our algorithms in Python and used Z3 as a back-end
SMT solver for solving the constraints over existential Presburger arithmetic. 
Our implementation takes as input a threshold automaton and a specification
in $\ELTLFT$ and checks if a counterexample exists.
We apply to the latest version of the benchmarks of~\cite{KW18}. Each benchmark yields two threshold automata, a hand-coded one and 
one obtained by a data abstraction of the algorithm written in Parametric Promela. For automata of the latter kind, due to data abstraction, we have to consider different cases for the same algorithm. We test each automaton against all specifications for that automaton.

Our tool outperforms ByMC in all automata with more than 30 states, with the exception of the second case of cbc. It performs worse in most small cases, however in these cases, both ByMC and our tool take less than 10 seconds.
ByMC works by enumerating all so-called \emph{schemas} of a threhold automaton, and
solving a SMT problem for each of them; the number of schemas can grow exponentially in the number of guards. Our tool avoids the enumeration. Since the number of schemas for the second case of cbc is just 2, while the second case of aba and third case of cf1s have more than 3000, avoiding the enumeration seems to be key to our better  performance. 


\section{Conclusions}

In this paper we have addressed the complexity of the most important verification and synthesis problems for threshold automata. In particular, we have shown that the coverability and reachability problems, as well as the model checking problem for the fault-tolerant temporal logic $\ELTLFT$ are all NP-complete, and that the bounded synthesis problem is $\Sigma_2^p$-complete. These results are a consequence of a novel characterization of the reachability relation of threshold automata as an existential formula of Presburger arithmetic. 


\bibliographystyle{plainurl}
\bibliography{References}
	
\clearpage

\appendix

\section*{APPENDIX}

\section{Detailed proofs from Section~\ref{sec:coverability}}

The intuition behind the proof of theorem~\ref{theorem:paramcov} has been
discussed in section~\ref{sec:coverability}. Here we provide the formal proof of the theorem.

The reduction for theorem~\ref{theorem:paramcov} is given by means of the
following definition, since the same construction will be used later on.

\begin{definition}\label{def:paramcov}
	Given a 3-CNF formula $\varphi = C_1 \land C_2 \land \dots \land C_m$ over variables $x_1,\dots,x_n$,
	we define its threshold automaton $\TA_\varphi=(\local, \initlocal, \varset, \ruleset)$ and the associated environment $\Env_\varphi$
	as follows: (see an illustrative example in Figure~\ref{fig:paramcov})
	\begin{itemize}
		\item $\Env_\varphi = (\paraset, \ResCond, \syssize)$ where
		$\paraset$ consists of a single parameter $k$, 
		$\ResCond$ is simply $\mathit{true}$ and $\syssize(k) = k$.
		\item The set of locations is $\local = \{\ell_1,\ldots,\ell_n, \top_1,\ldots,\top_n, \bot_1,\ldots,\bot_n, \ell_{\textit{mid}},\ell_F\}$. 
		\item The set of initial locations is $\iconfigs = \{\ell_1,\dots,\ell_n\}$.
		\item The variable set contains shared variables $\varset = \{y_1,\ldots,y_n, \bar{y_1},\ldots,\bar{y_n}, c_1,\ldots, c_m\}$.
		\item The set of rules~$\ruleset$ contains:
		\begin{itemize}
			\item For every $1 \le i \le n$, we have $(\ell_i,\top_i,\bar{y_i}<1,\cpp{y_i}) \in \ruleset$ and $(\ell_i,\bot_i,y_i<1,\cpp{\bar{y_i}}) \in \ruleset$.
			\item For every $1 \le i \le n$, we have $(\top_i,\ell_{\textit{mid}},\true,\cpp{c_{j_1}}\wedge\ldots\wedge\cpp{c_{j_s}}) \in \ruleset$
			where $C_{j_1},\ldots,C_{j_s}$ are all the clauses
			that $x_i$ appears in. 
			\item For every $1 \le i \le n$, we have 
			$(\bot_i,\ell_{\textit{mid}},\true,\cpp{c_{j_1}}\wedge\ldots\wedge\cpp{c_{j_s}}) \in \ruleset$
			where $C_{j_1},\ldots,C_{j_s}$ are all the clauses
			that $\neg x_i$ appears in. 
			\item $(\ell_{\textit{mid}},\ell_F,\phi,\zerovec) \in \ruleset$, 
			where  $\phi = c_1\ge 1\wedge \ldots \wedge c_m\ge 1$.
		\end{itemize}	
	\end{itemize}
\end{definition}

\emph{Remark: } Note that the constructed threshold automaton is acyclic and has only constant guards, i.e., only guards of the form
$x \ge a_0$ or $x < a_0$ for some $a_0 \in \NatZero$.

A simple consequence of the above definition is the following lemma.
It states that once some process moves to $\top_i$ (resp. $\bot_i$) no
process can move to $\bot_i$ (resp. $\top_i$).

\begin{lemma} \label{lem:noncont}
	Let $\varphi$ be a 3-CNF formula over variables $x_1,\dots,x_n$,
	and let its corresponding threshold automaton be $\TA_\varphi$.
	If $\sigma_0$ is an initial configuration and $\tau$ is a schedule 
	applicable to $\sigma_0$ then
	the following holds for every $i\in \{1,\ldots, n\}$ : If there exists $\sigma\in\setconf{\sigma_0}\tau$ with $\sigma.\counters[\top_i]\ge 1$
	$(\text{resp. }\sigma.\counters[\bot_i] \ge 1)$, then for every $\sigma'\in \setconf{\sigma_0}\tau$ it holds that $\sigma'.\counters[\bot_i]=0$, $(\text{resp. }\sigma'.\counters[\top_i]=0)$.
\end{lemma}

\begin{proof}
	Let $\infpath{\sigma_0}{\tau} = \sigma_0, t_1, \sigma_1, \dots, \sigma_{|\tau|-1},t_{|\tau|},\sigma_{|\tau|}$. Let $i \in \{1,\dots,n\}$ and let $\sigma_p$ be the first 
	configuration in the path such that $\sigma_p.\counters[\top_i] \ge 1$. 
	Hence $t_{p-1} = (\ell_i,\top_i,\bar{y_i}<1,\cpp{y_i})$.
	This means that $\sigma_{p-1}.\vars[\bar{y_i}]=0$ and so
	the rule
	$(\ell_i,\bot_i,y_i<1,\cpp{\bar{y_i}})$ could not have been fired
	before $\sigma_{p-1}$. Since this rule was never fired before $\sigma_{p-1}$ it follows that
	$\sigma_{q}.\counters[\bot_i] = 0$ for all $q \le p$.
	Also, notice that
	$\sigma_p.\vars[y_i] \ge 1$ and so the rule $(\ell_i,\bot_i,y_i<1,\cpp{\bar{y_i}})$ cannot be fired after $\sigma_p$. 
	Consequently we get that $\sigma_q.\counters[\bot_i] = 0$ for all $q > p$. The other claim is proven in a similar manner as well.
\end{proof}

\subsection{Proof of Theorem~\ref{theorem:paramcov}}

\begin{proof}
	We reduce from 3-SAT. Let $\varphi = C_1 \land C_2 \land \dots \land C_m$ be a 3-CNF formula over variables $x_1,\dots,x_n$ and consider
	the threshold automaton~$\TA_\varphi$ and the environment $\Env_\varphi$ as described in Definition~\ref{def:paramcov}.
	We show that $\varphi$ is satisfiable if and only if there is a path
	$\infpath{\sigma_0}\tau$ such that the configuration
	$\tau(\sigma_0)$ covers $\ell_F$.

	$(\Rightarrow)$	
	Suppose~$\varphi$ is satisfiable and~$\upsilon$ is a truth assignment that satisfies $\varphi$.
	We define an initial configuration~$\sigma_0$ and a schedule~$\tau$ 
	as follows.
	
	The initial configuration $\sigma_0$ has one process in 
	each of the initial locations $\ell_1,\dots,\ell_n$.
	As usual, all shared variables initially have value~0.
	
	The schedule~$\tau$ is defined as follows.
	For each~$i\in\{1,\dots,n\}$, if $\upsilon(x_i)=\true$ we fire the rule $(\ell_i,\top_i,\bar{y_i}<1,\cpp{y_i})$, and if $\upsilon(x_i)=\false$ then we fire the rule $(\ell_i,\bot_i,y_i<1,\cpp{\bar{y_i}})$.
	As all these guards initially evaluate to true, these~$n$ rules are applicable to~$\sigma_0$ (in any order), and lead to the (same) configuration~$\sigma_1$ such that $\sigma_1.\counters[\top_i]\ge1$ if and only if $\upsilon(x_i)=\true$ and $\sigma_1.\counters[\bot_i]\ge1$ if and only if $\upsilon(x_i)=\false$, and all the other locations are empty.
	Note that by construction, for each non-empty location from~$\sigma_1$, there is exactly one outgoing rule.
	Therefore, the rest of the schedule is formed in the only possible way, that is, each process moves to~$\ell_{\textit{mid}}$ and increments the corresponding shared variable(s) along that rule.
	
	Since each clause~$C_j$ evaluates to true under the assignment~$\upsilon$, every shared variable~$c_j$ is incremented at least once.
	Hence the guard of the rule leading to~$\ell_F$ becomes true, and so each process can move to~$\ell_F$, 
	thereby covering $\ell_F$.

	$(\Leftarrow)$
	Suppose there is an initial configuration~$\sigma_0$ and a schedule~$\tau$ such that $\tau(\sigma_0)=\sigma$ and $\sigma.\counters[\ell_F]\ge 1$.
	We construct the truth assignment~$\upsilon$ as follows: for each $i\in\{1,\ldots,n\}$ we define $\upsilon(x_i)=\true$ if $\top_i$ was visited along the path, i.e., if there is $\sigma'\in\setconf{\sigma_0}\tau$ with $\sigma'.\counters[\top_i]\ge 1$; similarly
	we define $\upsilon(x_i)=\false$,
	if there is $\sigma'\in\setconf{\sigma_0}\tau$ with $\sigma'.\counters[\bot_i]\ge 1$.
	If there is an $i$ such that neither $\top_i$ nor $\bot_i$ was visited by $\infpath{\sigma_0}\tau$, then we define $\upsilon(x_i)=\true$.
	Lemma~\ref{lem:noncont} implies that~$\upsilon$ is a well-defined truth assignment.
	By construction, it is easy to see that~$\upsilon$ satisfies~$\varphi$.
\end{proof}

\subsection{Proof of Theorem~\ref{theorem:nonparamcov}}

\begin{proof}
	Once again we give a reduction from 3-SAT.
	Let $\varphi$ be a 3-CNF formula over variables $x_1,\dots,x_n$. 
	We consider the threshold automaton $\TA_\varphi$ as described 
	in Definition~\ref{def:paramcov} and remove all the guards and updates
	on any rule that leaves any initial location.
	Formally, for every $i\in\{1,\ldots,n\}$, we replace the rule $(\ell_i,\top_i,\bar{y_i}<1,\cpp{y_i})$ with $(\ell_i,\top_i,\true,\zerovec)$. 
	Similarly, we replace $(\ell_i,\bot_i,y_i<1,\cpp{\bar{y_i}})$ 
	with $(\ell_i,\bot_i,\true,\zerovec)$.
	The rest of the rules in $\TA_\varphi$ are retained as they are.
	We will denote this new threshold automaton by $\TA'_{\varphi}$.
	
	We define the initial configuration~$\sigma_0$ such that there is exactly one process in each initial state $\ell_i$ (as usual, all other locations are empty and all shared variables have values~0).
	
	$(\Rightarrow)$ 
	Suppose $\varphi$ can be satisfied with assignment $\upsilon$.
	Notice that in this case, we already described a path from
	$\sigma_0$ covering $\ell_F$ in the previous proof.
	
	$(\Leftarrow)$
	Suppose $\ell_F$ can be covered from $\sigma_0$. 
	Let $\infpath{\sigma_0}{\tau}$ be a path along which $\ell_F$ is covered.
	Define $\upsilon$ to be the following truth assignment for $\varphi$:
	$\upsilon(x_i) = \true$ if the rule $(\ell_i,\top_i,\true,\zerovec)$
	was fired in the schedule $\tau$,
	$\upsilon(x_i) = \false$ if the rule $(\ell_i,\bot_i,\true,\zerovec)$
	was fired in the schedule $\tau$, 
	and $\upsilon(x_i) = \true$ if neither of these rules were fired.
	Since the automaton is acyclic and there is only one process
	at $\ell_i$, 
	both $(\ell_i,\top_i,\true,\zerovec)$ and $(\ell_i,\bot_i,\true,\zerovec)$ cannot be fired and so this
	is a well-defined assignment.
	It can then be easily verified that
	$\upsilon$ is a satisfying assignment for $\varphi$.
\end{proof}

\subsection{Proof of Theorem~\ref{theorem:cons-rise-complete}}

\begin{proof}

For proving P-hardness, we give a logspace reduction from the Circuit Value Problem (CVP) which
is known to be P-complete \cite{CVP}.
The reduction we present is similar to the reduction given in Proposition 1 of ~\cite{Mobile}.

CVP is defined as follows: Given a boolean circuit~$C$ with~$n$ input variables and~$m$ gates, and a truth assignment $\upsilon$ for the input variables, check if~$\upsilon$ evaluates to~1 on the circuit~$C$. 

We represent each binary gate~$g$ as a tuple $(\circ,s_1,s_2)$, where $\circ\in\{\wedge,\vee\}$ denotes the operation of~$g$, and $s_1,s_2\in\{x_1,\ldots,x_n,g_1,\ldots,g_m\}$ are the inputs to $g$.
In a similar fashion, each unary gate~$g$ is represented as a tuple $(\neg,s)$. 
By convention, $g_1$ is always the output gate.

Based on this we present the following reduction:
(See Figure~\ref{fig:risecov} for an illustrative example).
Let~$C$ be a boolean circuit with input variables $x_1,\dots,x_n$ and  gates $g_1,\ldots,g_m$, and let~$\upsilon$ be a truth assignment that for every $x_i$, assigns $\upsilon(x_i)=b_i\in \{0,1\}$.
We define the threshold automaton $\TA_{C,v}=(\local, \initlocal, \varset, \ruleset)$ and the environment $\Env_{C,v} = (\paraset, \ResCond, \syssize)$ as follows:
\begin{itemize}
	\item $\Env_{C,v} = (\paraset, \ResCond, \syssize)$ where
	$\paraset$ consists of a single parameter $k$, $\ResCond$ is 
	simply true and $\syssize(k) = k$.
	\item The set of locations is $\local = \{\ell_0,\ell_1,,\ldots,\ell_n, \ell_{g_1},\ell_{g_1}',\ldots, \ell_{g_m},\ell_{g_m}', \ell_F\}$,
	\item The set of initial locations is $\initlocal = \{\ell_0,\ell_{g_1},\ell_{g_2},\dots,\ell_{g_m}\}$,
	\item  The set of shared variables is $\varset = \{x_1^0, x_1^1,\ldots, x_n^0, x_n^1, g_1^0,g_1^1,\ldots,g_m^0,g_m^1\}$
	\item The set of rules~$\ruleset$ contains:
	\begin{itemize}
		\item a rule
		$(\ell_{i-1},\ell_i,\true,\cpp{x_i^{b_i}})$, for every $1 \le i \le n$
		\item a rule $(\ell_{g_1}',\ell_{F},g_1^1\ge 1,\zerovec)$, where $g_1$ is the output gate,
		\item rules $(\ell_{g},\ell_{g}',s_1^{b_1}\ge 1 \wedge s_2^{b_2}\ge 1,\cpp{g^{b_1\circ b_2}})$, for every binary gate $g = (\circ,s_1,s_2)$, and for every $(b_1,b_2)\in\{0,1\}^2$,
		\item rules $(\ell_{g},\ell_{g}',s^{b}\ge 1,\cpp{g^{\neg b}})$, for every unary gate $g = (\neg,s)$, and for every $b\in\{0,1\}$.
	\end{itemize}	 
\end{itemize}

It can be easily checked that if $\sigma_0$ is an initial configuration and $\tau$ is a schedule such that $\tau(\sigma_0) = \sigma$
then $\sigma.\vars[x_i^b] \ge 1$ iff $b = b_i$ and
$\sigma.\vars[g^b] \ge 1$ iff the gate $g$ outputs $b$ when the
input variables are given the assignment $\upsilon$.
Consequently, the state~$\ell_F$ can be covered if and only if the circuit evaluates to 1 on the given assignment. 
Since the reduction can be clearly accomplished in logspace it follows that the coverability problem for this case is P-hard.

\begin{figure}
	\tikzstyle{edgelabel}=[sloped, above, align= center]
	\tikzstyle{node}=[circle,draw=black,thick,minimum size=4.3mm,inner sep=0.75mm,font=\normalsize]
	
	\begin{tikzpicture}[->,node distance = 2cm, scale=0.7, every node/.style={scale=0.8}]	
	
	\node[node] (l0)[label=below:\textcolor{blue}{$\ell_0$}] {};
	\node[node, right = of l0] (l1)[label=below:\textcolor{blue}{$\ell_1$}] {};
	\node[node, right = of l1] (l2)[label=below:\textcolor{blue}{$\ell_2$}] {};
	\node[node, right = of l2] (l3)[label=below:\textcolor{blue}{$\ell_3$}] {};
	
	\draw(l0) edge[edgelabel] node{$\cpp{x_1^{1}}$} (l1);
	\draw(l1) edge[edgelabel] node{$\cpp{x_2^{1}}$} (l2);
	\draw(l2) edge[edgelabel] node{$\cpp{x_3^{1}}$} (l3);
	
	\node[node, below = 1cm of l0] (g1)[label=below:\textcolor{blue}{$\ell_{g_1}$}] {};
	\node[node, right = 7.5cm of g1] (g1')[label=below:\textcolor{blue}{$\ell_{g_1}'$}] {};
	\node[node, right = of g1'] (f) [label=below:\textcolor{blue}{$\ell_{F}$}]{};
	
	\draw(g1) edge[edgelabel] node{$\forall b,b' \in \{0,1\}^2, \ x_1^b\ge 1 \wedge g_2^{b'} \ge 1 \mapsto \cpp{g_1^{b \land b'}}$} (g1');
	\draw(g1') edge[edgelabel] node{$g_1^1 \ge 1$} (f);
	
	\node[node, below = 1cm of g1] (g2)[label=below:\textcolor{blue}{$\ell_{g_2}$}] {};
	\node[node, right = 7.5cm of g2] (g2')[label=below:\textcolor{blue}{$\ell_{g_2}'$}] {};
	
	\draw(g2) edge[edgelabel] node{$\forall b,b' \in \{0,1\}^2, \ g_3^b\ge 1 \wedge x_3^{b'} \ge 1 \mapsto \cpp{g_2^{b \lor b'}}$} (g2');
	
	\node[node, below = 1cm of g2] (g3)[label=below:\textcolor{blue}{$\ell_{g_3}$}] {};
	\node[node, right = 7.5cm of g3] (g3') [label=below:\textcolor{blue}{$\ell_{g_3}'$}] {};
	
	\draw(g3) edge[edgelabel] node{$\forall b \in \{0,1\}, \ x_2^b \ge 1 \mapsto \cpp{g_3^{\lnot b}}$} (g3');
	
	\end{tikzpicture}
	\caption{Threshold automaton~$\TA_{C,v}$ corresponding to the circuit $x_1 \wedge (\neg x_2 \vee x_3)$ with gates $g_1 = (\wedge, x_1, g_2)$, $g_2 = (\vee, g_3, x_3)$, $g_3 = (\neg, x_2)$ and assignment
		$\upsilon(x_1) = 1, \upsilon(x_2) = 1, \upsilon(x_3) = 1$.
		(For brevity, we have denoted all the four transitions possible
		between $\ell_{g_1}$ and $\ell'_{g_1}$ with a single transition
		quantifying over all the four possible choices. Similar notation is 
		employed for the other transitions as well.)
		The assignment~$\upsilon$ is a satisfying assignment of the circuit~$C$, and it is easy to see that 
		if the initial configuration has one process in each of the initial locations $\ell_0, \ell_{g_1},\ell_{g_2}, \ell_{g_3}$, then there is (only) one possible path from that configuration, and it covers~$\ell_F$.
	}\label{fig:risecov}
\end{figure}

\paragraph{A polynomial time algorithm: }

	Let $\TA = (\local, \initlocal, \varset, \ruleset)$ be a threshold automaton such that each transition in $\ruleset$ only has constant rise guards (i.e., guards of the form $x \ge a_0$ for some $a_0 \in \NatZero$). 
	Let $\Env = (\paraset, \ResCond, \syssize)$ be an environment
	which is multiplicative (See definition ~\ref{def:mult}).
	Clearly a guard of the form $x \ge 0$ is redundant and so
	we assume that if $x \ge c$ is a guard in some rule in $\TA$, then $c > 0$. 
	
	Given two configurations $\sigma = (\counters,\vars,\param)$ and $\sigma' = (\counters',\vars',\param')$ let
	$\sigma + \sigma'$ be the configuration $(\counters+\counters',\vars+\vars',\param+\param').$
	Similarly given $c \in \NatZero$ let $c \cdot \sigma$ be the configuration $(c \cdot \counters, c \cdot \vars, c \cdot \param)$.
	Since the automaton has only constant rise guards it follows that
	\begin{equation} \label{eq:important}
		r \text{ is enabled at } \sigma \implies r \text{ is enabled at }
		\sigma + \sigma' \text{ for any configuration } \sigma'
	\end{equation}
	
	We prove that coverability in this case is in P by means of a saturation algorithm which begins with the set of initial locations and finishes with the set of all possible coverable locations. 
	This algorithm is similar to the one mentioned in~\cite{Mobile}.
	
	We initialize two variables $X_\local$ and $X_\ruleset$ by $X_\local := \initlocal$ and $X_\ruleset := \emptyset$, and repeatedly update 
	them until a fixed point is reached, as follows:
	\begin{itemize}
		\item If there exists a rule $r = (\ell,\ell',\true,\vec{u}) \in \ruleset \setminus X_\ruleset$ such that $\ell \in X_\local$, then set $X_\local := X_\local \cup \{\ell'\}$ and $X_\ruleset := X_\ruleset \cup \{r\}$.
		\item If there exists a rule $r=(\ell,\ell', (\land_{1 \le i \le q} \ x_i \geq c_i), \vec{u}) \in \ruleset \setminus X_\ruleset$ such that $\ell \in X_\local$, \emph{and} there exists rules $r_1,r_2,\dots,r_q$ such that each
		$r_i = (\ell_i,\ell'_i,\varphi_i,\vec{u}_i) \in X_\ruleset$
		and $\vec{u}_i[x_i] > 0$, then set $X_\local := X_\local \cup \{\ell'\}$ and $X_	\ruleset := X_\ruleset \cup \{r\}$.
		
	\end{itemize}
	
	By means of this algorithm we get a sequence of 
	sets $(X_\local^0,X_\ruleset^0) \subseteq (X_\local^1,X_\ruleset^1) \subseteq \dots \subseteq
	(X_\local^m,X_\ruleset^m)$, one for each iteration. 
	We say that a rule $r \in \ruleset$ can possibly occur if 
	there is some initial configuration $\sigma_0$ such that
	$\sigma_0 \xrightarrow{*} \sigma$ and $\sigma$ enables $r$.
	We will now show that $X_\local^m$ contains exactly the set
	of coverable locations and $X_\ruleset^m$ contains exactly the
	set of rules which can possibly occur.


	First we show by induction on $i$, that if $\ell \in X_\local^i$ then $\ell$ is coverable and if $r \in X_\ruleset^i$ then $r$ can possibly occur. The claim is clearly true for the base
	case of $i = 0$. Suppose the claim is true for some $i$.
	Let $X_\local^{i+1} \setminus X_\local^i = \{\ell'\}$
	and $X_\ruleset^{i+1} \setminus X_\ruleset^i = \{r\}$.
	Let $c$ be the maximum constant appearing in any of the guards of any 
	of the rules of $\TA$.
	
	By construction, there are two possible cases:
	\begin{itemize}
		\item $r$ is of the form $r = (\ell,\ell',true,\vec{u})$ such that $\ell \in X_\local^i$. By inductive hypothesis $\ell$
		is coverable. It is then clear that $\ell'$ is coverable 
		and also that $r$ can possibly occur.
		\item $r$ is of the form $r = (\ell,\ell',(\land_{1 \le j \le q}  \ x_j \ge c_j), \vec{u})$ such that $\ell \in X_\local^i$. 
		We have to show that 
		\begin{equation} \label{eq:toshow}
			\ell' \text{ is coverable and } r \text{ can possibly occur}
		\end{equation}
		By construction of the algorithm, for each $j \in \{1,\dots,q\}$, there
		exists $r_j = (\ell_j,\ell'_j,\varphi_j,\vec{u}_j)\in X_\ruleset^i$ such that $\vec{u}_j[x_j] > 0$. By inductive hypothesis, $\ell$ is coverable and so there exists 
		$\sigma_0^{\ell} \xrightarrow{*} \sigma^{\ell}$ such that $\sigma.\counters[\ell] \ge 1$. Once again, by inductive hypothesis, for each $j$, the rule $r_j$ can possibly occur and so there exists
		$\sigma_0^j \xrightarrow{*} \sigma^j$ such that $\sigma^j$
		enables $r_j$. 
		
		Let $\sigma_0 = c \cdot (\sigma_0^{\ell} + \sigma_0^1 + \dots + \sigma_0^q)$ and let $\sigma_{mid} = (c \cdot (\sigma^{\ell} + \sigma^1 + \dots + \sigma^q))$.
		Since the environment is multiplicative, $\sigma_0$ is
		a valid initial configuration. It is then easy to see (by repeated applications of observation~\ref{eq:important})
		that $\sigma_0 \xrightarrow{*} \sigma_{mid}$. Notice that, by observation~\ref{eq:important},
		$r_1,\dots,r_q$ are all enabled at $\sigma_{mid}$
		and also that $\sigma_{mid}[\ell_1] \ge c, \dots \sigma_{mid}[\ell_q] \ge c$.
		By firing each of $r_1,\dots,r_q$ exactly $c$ many times, we 
		can arrive at a configuration $\sigma$ where $r$ is enabled. 
		Now firing $r$ from $\sigma$ covers $\ell'$. 
		Hence $\ell'$ is coverable and $r$ can possibly occur thereby proving~(\ref{eq:toshow}).
		
	\end{itemize}
	
	For the converse direction, suppose
	$\sigma_0$ is an initial configuration such that 
	$\sigma_0 \xrightarrow{*} \sigma$ and $\sigma$ covers some
	location $\ell$ and enables some rule $r$. By an 
	easy induction on the length of the path between $\sigma_0$
	and $\sigma$ we can establish that $\ell \in X_\local^m$
	and $r \in X_\ruleset^m$.

\end{proof}

\section{Detailed proofs from Section~\ref{sec:reachability}}

In the following, when we say that $\phi_\mathit{steady}(\sigma,\sigma')$ holds with assignment
$\{y_r\}_{r \in \ruleset}$ (where each $y_r$ is a natural number),
we mean that the formula $\phi(\sigma,\sigma')$ is true when each 
variable $x_r$ is assigned the value $y_r$.
Given two assignments $Y = \{y_r\}_{r \in \ruleset}$ and
$Z = \{z_r\}_{r \in \ruleset}$ we say that $Z < Y$ iff 
$\sum_{r \in \ruleset} z_r < \sum_{r \in \ruleset} y_r$.
Further for an assignment $Y = \{y_r\}_{r \in \ruleset}$ and a location
$\ell$ we say that there is a \emph{fireable cycle} at $\ell$ with respect to $Y$ iff there exists rules $T = \{t_1,\dots,t_m\}$ such that
$y_{t_1} > 0, \dots, y_{t_m} > 0$, $t_i.\tostate = t_{i+1}.\fromstate$ 
for all $i < m$,
and $t_m.\tostate = t_1.\fromstate = \ell$.

The following two lemmas are important properties of the formula
$\phi_\mathit{steady}$.

\begin{lemma} \label{lem:nocycles}
	Let $\phi_\mathit{steady}(\sigma,\sigma')$ be true with the assignment
	$Y = \{y_r\}_{r \in \ruleset}$. Suppose $y_t > 0$ and $\sigma[t.\fromstate] > 0$ for some rule $t$. If there are no fireable cycles at $t.\fromstate$ with respect to $Y$, then $\phi_\mathit{steady}(t(\sigma),\sigma')$ is true with an
	assignment $Z$ where $Z < Y$.
\end{lemma}

\begin{lemma} \label{lem:cycles}
	Let $\phi_\mathit{steady}(\sigma,\sigma')$ be true with the assignment
	$Y = \{y_r\}_{r \in \ruleset}$. Suppose $\sigma[\ell] > 0$ and suppose there is a fireable cycle $\{t_1,\dots,t_m\}$
	at $\ell$ with respect to $Y$. Then $\phi_\mathit{steady}(t_1(\sigma),\sigma')$ is true with an
	assignment $Z$ where $Z < Y$.
\end{lemma}

We will first see how using these two lemmas, we can prove
Theorem~\ref{th:kirchoffpath}. Then we will present the proofs of these 
two lemmas.

\subsection{Proof of Theorem~\ref{th:kirchoffpath}}

\begin{proof}
	It is clear from our discussion during the construction
	of the formula $\phi_\mathit{steady}$
	that if there is a steady run from~$\sigma$ to~$\sigma'$ then $\phi_\mathit{steady}(\sigma,\sigma')$ is true. 
	Hence, we will only show that this is a sufficient condition.
	
	Let $\phi_\mathit{steady}(\sigma,\sigma')$ be true with the assignment $\{y_r\}_{r \in \ruleset}$.
	We induct on the value of $\sum_{r \in \ruleset} y_r$.
	
	\textsc{Base Case. }
	It can be easily verified that $\forall r \in \ruleset. \  y_r = 0$ if and only if $\sigma = \sigma'$. 
	Hence, if all $y_r$ are zero, then $\sigma = \sigma'$ and so we are done. 
	This constitutes the base case of the induction.
	
	\textsc{Induction Hypothesis. }
	Let us fix values $Y = \{y_r\}_{r\in\ruleset}$, such that not all of them are zero.
	We assume that if $\phi_\mathit{steady}(c,c')$ holds for some configurations~$c$ and~$c'$ with the assignment $Z = \{z_r\}_{r \in \ruleset}$ such that $Z < Y$, then there is a steady schedule~$\tau$ such that $\tau(c)=c'$.
	
	\textsc{Induction Step. }
	Assume now that $\phi_\mathit{steady}(\sigma,\sigma')$ holds with the  assignment $Y = \{y_r\}_{r\in\ruleset}$.
	We show here that there exists a steady schedule~$\tau$ such that $\tau(\sigma)=\sigma'$, using the following idea: 
	(i) we construct a configuration~$\sigma''$ that is reachable from $\sigma$ in one step, i.e., we show that there is a transition~$t$ such that $t(\sigma)=\sigma''$, and 
	(ii) we prove that $\phi_\mathit{steady}(\sigma'',\sigma')$ holds for an assignment
	$\{z_r\}_{r\in\ruleset}$ such that $\{z_r\}_{r\in\ruleset} < \{y_r\}_{r\in\ruleset}$. 
	By applying our induction hypothesis, there is a schedule~$\rho$ with $\rho(\sigma'')=\sigma'$. 
	We would then have $\tau = t\cdot \rho$ as the required schedule.
	
	Let us construct~$\sigma''$.
	By the assumption, there is a rule $r'\in\ruleset$ such that $y_{r'}>0$.
	As $\phi_\mathit{steady}(\sigma,\sigma')$ holds with the assignment $Y=\{y_r\}_{r\in\ruleset}$, $\phi_{\textit{appl}}(\sigma,Y)$ is true and therefore there must be a set of rules $S=\{r_1,\dots,r_s\} \subseteq \ruleset$ such that $\phi_\mathit{chain}^{r'}(S,\sigma,Y)$ holds, that is,
	$$
	\sigma.\counters[r_1.\fromstate] > 0 \ \land \
	\bigwedge_{1\le i\le s} y_{r_i} > 0 \ \land \ \bigwedge_{1 < i \le s} r_{i-1}.\tostate = r_i.\fromstate \ \land \ r_{s} = r'
	$$
	
	Let $r = r_1$ and let $\ell = r.\fromstate$. If 
	there is no fireable cycle at $\ell$ with respect to $Y$, then by Lemma~\ref{lem:nocycles} 
	we have that $\phi_\mathit{steady}(r(\sigma),\sigma')$ is true
	with an assignment $Z < Y$ and so we can define $\sigma'' = r(\sigma)$. Otherwise, if $\{t_1,\dots,t_m\}$ is a fireable
	cycle at $\ell$ then by Lemma~\ref{lem:cycles} we have that
	$\phi_\mathit{steady}(t_1(\sigma),\sigma')$ is true with 
	an assignment $Z < Y$ and so we can define $\sigma'' = t_1(\sigma)$.
	Hence we can construct the required
	$\sigma''$ in either case and so the proof is complete.
	
\end{proof}

Now we will prove the Lemmas~\ref{lem:nocycles} and~\ref{lem:cycles}.

\subsection{Proof of Lemma~\ref{lem:nocycles}}

\begin{proof}
	Let $\ell := t.\fromstate$.
	Since $y_t > 0$, from 
	$\phi_\ruleset(\sigma,Y)$ we obtain that $\sigma \models t.\varphi$. 
	As we have $\sigma.\counters[t.\fromstate] = \sigma.\counters[\ell] > 0$ we know that 
	$t$ is enabled at $\sigma$.
	Hence we set $\sigma''=t(\sigma)$ and 
	let $Z = \{z_r\}_{r \in \ruleset}$ be the following assignment:
	$z_r = y_r$ if $r \neq t$ and $z_r = y_r - 1$ if $r = t$.
	It can be easily checked that all of the following holds:
	$\phi_\mathit{base}(\sigma'',\sigma') \land
	\phi_\local(\sigma'',\sigma',Z) \land
	\phi_\globset(\sigma'',\sigma',Z) \land 
	\phi_\ruleset(\sigma'',Z) $. 
	Therefore, we only need to prove that $\phi_{\textit{appl}}(\sigma'',Z)$ holds.
	
	Fix any arbitrary rule $r \in \ruleset$ such that $z_r > 0$.
	We now have to show that there exists a set of rules $S'=\{r'_1,\ldots,r'_{s'}\}\subseteq \ruleset$ such that
	$$
	\sigma''.\counters[r'_1.\fromstate] > 0 \ \land \
	\bigwedge_{1\le i\le s'} z_{r'_i} > 0 \ \land \ \bigwedge_{1 < i \le s'} r'_{i-1}.\tostate = r'_{i}.\fromstate \ \land \ r'_{s'} = r
	$$
	
	Since $z_r > 0$, it follows that $y_r > 0$.
	Since $\phi_{\textit{appl}}(\sigma,Y)$ holds, it follows that
	there exists a set of rules $S = \{r_1,\dots,r_s\} \subseteq \ruleset$ such that
	$$
	\sigma.\counters[r_1.\fromstate] > 0 \ \land \
	\bigwedge_{1\le i\le s} y_{r_i} > 0 \ \land \ \bigwedge_{1 < i \le s} r_{i-1}.\tostate = r_{i}.\fromstate \ \land \ r_{s} = r
	$$
	We consider three cases:
	\begin{itemize}
		\item Suppose $\sigma''.\counters[r_1.\fromstate] > 0$ and
		suppose $z_{r_i} > 0$ for every $1\le i\le s$. Then we can
		simply set $S' := S$.
		\item Suppose there exists $i$ such that $z_{r_i} = 0$. 
		Hence $z_{r_i} \neq y_{r_i}$ and this immediately implies that $r_i = t$. 
		Since $\sigma''$ was obtained from $\sigma$ by 
		firing $t$, it follows that $\sigma''.\counters[t.\tostate] = \sigma''.\counters[r_i.\tostate] = \sigma''.\counters[r_{i+1}.\fromstate] > 0$. 
		It is then easy to see that we can 
		set $S' := \{r_{i+1},\dots,r_s\}$.
		\item Suppose $\sigma''.\counters[r_1.\fromstate] = 0$.
		We can assume that $z_{r_i} > 0$ for every $1 \le i \le s$  as otherwise the previous case
		applies.
		Since $\sigma''$ was obtained from $\sigma$ by firing
		$t$, it follows that $r_1.\fromstate = t.\fromstate = \ell$. Let $in^{\ell}_1,\dots,in^{\ell}_a$ and
		$out^{\ell}_1,\dots,out^{\ell}_b$ be the set of
		incoming and outgoing rules at $\ell$ in the  threshold automaton.
		We further split this into two more cases:
		\begin{itemize}
			\item There exists $i$ such that $z_{in^{\ell}_i} > 0$. Let $r'' = in^{\ell}_i$. Since $z_{r''} > 0$ it follows
			that $y_{r''} > 0$. Hence there exists a set $S'' = \{r''_1,\dots,r''_{s''}\} \subseteq \ruleset$ such that
			$$
			\sigma.\counters[r''_1.\fromstate] > 0 \ \land \
			\bigwedge_{1\le i\le s''} y_{r''_i} > 0 \ \land \ \bigwedge_{1 < i \le s''} r''_{i-1}.\tostate = r''_{i}.\fromstate \ \land \ r''_{s''} = r''
			$$
			Since $\ell$ is not part of any fireable cycle with respect to 
			$Y$ and since $r''.\tostate = \ell$ it follows that 
			$r''_1.\fromstate  \neq \ell$. 
			Hence $\sigma.\counters[r''_1.\fromstate] > 0
			\implies \sigma''.\counters[r''_1.\fromstate] > 0$.
			Once again, since $\ell$ is not part of any
			fireable cycle with respect to $Y$, it follows
			that $r''_i \neq t$ for any $i$. Therefore
			$y_{r''_i} > 0 \implies z_{r''_i} > 0$.
			Hence we can
			set $S' := \{r''_1,\dots,r''_{s''},r_1,\dots,r_s\}$.
			\item There does not exist $i$ such that  $z_{in^{\ell}_i} > 0$. We show that this case
			cannot happen.
			Recall that the 
			assignment $Z$ satisfies the formula $\phi_\local(\sigma'',\sigma',Z)$. In particular, 
			$$
			\sum_{i=1}^a z_{in^{\ell}_i} - \sum_{i=1}^b z_{out^{\ell}_i} = \sigma'.\counters[\ell] - \sigma''.\counters[\ell]
			$$
			
			By assumption, $z_{in^{\ell}_i} = 0$ for all $i$, $1 \le i \le a$. Also $z_{r_1} > 0$ and $r_1$ is (by definition) an outgoing transition
			from $l$. Hence the left hand side of the equation is strictly less
			than 0. Since $\sigma''.\counters[\ell] = 0$, it follows that
			the right hand side of the equation is bigger than or equal to
			0, which gives rise to a contradiction.
		\end{itemize}
	\end{itemize}
	
\end{proof}

\subsection{Proof of Lemma~\ref{lem:cycles}}

\begin{proof}
	Suppose $\ell$ is part of some fireable cycle with respect to $Y$, i.e., 
	there exists rules $T = \{t_1,\dots,t_m\}$ such that
	$y_{t_1} > 0, \dots, y_{t_m} > 0$, $t_i.\tostate = t_{i+1}.\fromstate$
	for all $i < m$, and $t_m.\tostate = t_1.\fromstate = \ell$.
	Let $t = t_1$. 
	Since $y_t > 0$, from 
	$\phi_\ruleset(\sigma,Y)$ we obtain that $\sigma \models t.\varphi$. 
	As we have $\sigma.\counters[t.\fromstate] = \sigma.\counters[\ell] > 0$ we know that 
	$t$ is enabled at $\sigma$.
	Hence we set $\sigma''=t(\sigma)$ and 
	let $Z = \{z_r\}_{r \in \ruleset}$ be the following assignment:
	$z_r = y_r$ if $r \neq t$ and $z_r = y_r - 1$ if $r = t$.
	It can be easily checked that all of the following holds:
	$\phi_\mathit{base}(\sigma'',\sigma') \land
	\phi_\local(\sigma'',\sigma',Z) \land
	\phi_\globset(\sigma'',\sigma',Z) \land 
	\phi_\ruleset(\sigma'',Z) $. 
	Therefore, we only need to prove that $\phi_{\textit{appl}}(\sigma'',Z)$ holds.

	Fix any arbitrary rule $r \in \ruleset$ such that $z_r > 0$.
	We now have to show that there exists a set of rules $S'=\{r'_1,\ldots,r'_{s'}\}\subseteq \ruleset$ such that
	$$
	\sigma''.\counters[r'_1.\fromstate] > 0 \ \land \
	\bigwedge_{1\le i\le s'} z_{r'_i} > 0 \ \land \ \bigwedge_{1 < i \le s'} r'_{i-1}.\tostate = r'_{i}.\fromstate \ \land \ r'_{s'} = r
	$$
	
	Since $z_r > 0$, it follows that $y_r > 0$.
	Since $\phi_{\textit{appl}}(\sigma,Y)$ holds, it follows that
	there exists a set of rules $S = \{r_1,\dots,r_s\} \subseteq \ruleset$ such that
	$$
	\sigma.\counters[r_1.\fromstate] > 0 \ \land \
	\bigwedge_{1\le i\le s} y_{r_i} > 0 \ \land \ \bigwedge_{1 < i \le s} r_{i-1}.\tostate = r_{i}.\fromstate \ \land \ r_{s} = r
	$$
	We now consider three cases:
	\begin{itemize}
		\item Suppose $\sigma''.\counters[r_1.\fromstate] > 0$ and
		suppose $z_{r_i} > 0$ for every $1\le i\le s$. Then we can
		simply set $S' := S$.
		\item Suppose there exists $i$ such that $z_{r_i} = 0$. 
		Hence $z_{r_i} \neq y_{r_i}$ and this immediately implies that $r_i = t$. 
		Since $\sigma''$ was obtained from $\sigma$ by 
		firing $t$, it follows that $\sigma''.\counters[t.\tostate] = \sigma''.\counters[r_i.\tostate] = \sigma''.\counters[r_{i+1}.\fromstate] > 0$. 
		It is then easy to see that we can 
		set $S' := \{r_{i+1},\dots,r_s\}$.
		\item Suppose $\sigma''.\counters[r_1.\fromstate] = 0$.
		We can assume that $z_{r_i} > 0$ for every $1 \le i \le s$  as otherwise the previous case
		applies.
		Since $\sigma''$ was obtained from $\sigma$ by firing
		$t$, it follows that $r_1.\fromstate = t.\fromstate = \ell = t_m.\tostate$. Notice 
		that $\sigma''.\counters[t.\tostate] = 
		\sigma''.\counters[t_2.\fromstate] > 0$ and also that $z_{t_i} = y_{t_i} > 0$ for all $i > 1$. Hence we can set 
		$S' := \{t_2,\dots,t_m,r_1,\dots,r_s\}$.
	\end{itemize}
\end{proof}

\emph{Remark: } Notice that, in the course of 
proving Theorem~\ref{th:kirchoffpath}, we have actually proved something much stronger. 
For a schedule $\tau$ and a rule $r$, let $\tau(r)$ denote the
number of times $r$ is present in the schedule $\tau$.
Analysing our proof of Theorem~\ref{th:kirchoffpath}, it is easy
to notice the following:
\begin{lemma}\label{lem:observation}
	If $\tau$ is a steady
	schedule between $\sigma$ and $\sigma'$ then $\phi_\mathit{steady}(\sigma,\sigma')$
	holds with assignment $\{\tau(r)\}_{r \in \ruleset}$.
	Further, if $\phi_\mathit{steady}(\sigma,\sigma')$ holds
	with assignment $\{y_r\}_{r \in \ruleset}$ then it is possible
	to create a steady schedule $\tau$ between $\sigma$ and $\sigma'$ such 
	that $\tau(r) = y_r$ for all rules $r$. 
\end{lemma}

\subsection{Proof of Theorem~\ref{th:main}}

\begin{proof}
	Let $\sigma$ and $\sigma'$ be two configurations.
	First we will show that deciding if there is a \emph{steady schedule}
	$\tau$ such that~$\tau(\sigma) = \sigma'$ can be done in NP.
	From Theorem~\ref{th:kirchoffpath} we know there is a steady schedule from $\sigma$ to $\sigma'$ if and only if
	the existential Presburger formula
	$\phi_\mathit{steady}(\sigma,\sigma')$ is true. 
	However the formula $\phi_\mathit{steady}(\sigma,\sigma')$ is exponential in size when constructed naively. 
	But the exponential dependence comes only from the constraint $\phi_{\textit{appl}}$ which can be easily reduced to a polynomial dependence as follows:
	For each rule~$t$ we guess a set $S = \{r_1^t,\dots,r_s^t\}$ and check that $r_s^t = t$ and $r_{i}^t.\fromstate = r_{i+1}^t.\tostate$ for all $i < s$.
	Once that is done, we replace $\phi_{\textit{appl}}$ with
	$$
	\bigwedge_{t \in R} x_t > 0 \implies  
	\sigma[r_1^t.from] > 0 \ \land \
	x_{r_1^t} > 0 \ \land \dots \land \ x_{r_s^t} > 0 
	$$
	
	It is then clear that $\phi_\mathit{steady}(\sigma,\sigma')$ is true if and only if at least one of our guesses is true. 
	Now we have reduced it to a polynomial sized existential Presburger arithemtic formula which we can decide in NP (See~\cite{ExistPres}).
	This shows that checking whether there is a steady schedule
	from $\sigma$ to $\sigma'$ is in NP.
	
	Now suppose we want to check if there is a (general) run from $\sigma$ to $\sigma'$. By Proposition~\ref{prop:gen-to-steady},
	every path can be written as the concatenation of 
	at most $K$ steady paths where $K = |\PrecondU| + |\PrecondL| + 1$. 
	Hence it suffices to check if there is a run of the form
	\begin{equation*}
		\sigma = \sigma_0 \xrightarrow{*} \sigma'_0 \rightarrow \sigma_1 \xrightarrow{*} \sigma'_1 \rightarrow \sigma_2 \dots \sigma_{K} \xrightarrow{*} \sigma'_{K} = \sigma'
	\end{equation*}
	such that the context of each $\sigma_i$ is the same as the context of 
	$\sigma'_i$. With this in mind, we define $\phi_\mathit{step}(\eta,\eta')$
	as the formula $\phi_\mathit{steady}(\eta,\eta')$, except we do not enforce
	that $\eta$ and $\eta'$ have the same context and we enforce
	that the existential variables $\{x_r\}_{r \in \ruleset}$ obey the constraint $\sum_{r \in \ruleset} x_r \le 1$. It is
	easy to see that
	$\phi_\mathit{step}(\eta,\eta')$ is true iff $\eta'$ can be
	reached from $\eta$ in at most one step.
	Hence if we define $\phi_\mathit{reach}(\sigma,\sigma')$ to be 
	\begin{equation} \label{eq:reach}
		\exists \ \sigma_0, \sigma'_0, \dots, \sigma_{K}, \sigma'_{K} \
		\left(\sigma_0 = \sigma \land \sigma'_{K} = \sigma' \land \bigwedge_{0 \le i \le K} \phi_\mathit{steady}(\sigma_i,\sigma'_i) 
		\land  \bigwedge_{0 \le i \le K-1}
		\phi_\mathit{step}(\sigma'_i,\sigma_{i+1}) \right)
	\end{equation}
	then it is clear that there is a run from $\sigma$ to $\sigma'$ 
	iff $\phi_\mathit{reach}(\sigma,\sigma')$ is satisfied. 
	To decide if $\phi_\mathit{reach}(\sigma,\sigma')$ is true, we eliminate the exponentially sized disjunctions in $\phi_\mathit{steady}$ as before 
	and check that the resulting (polynomial sized) formula is satisfiable.
\end{proof}

The proof of Theorem~\ref{th:main} is complete. 
However, for future purposes, we modify the formula $\phi_{reach}(\sigma,\sigma')$ slightly,
so that it holds some extra information. 
Notice that in equation~\ref{eq:reach},
for each rule $r$, there is an existential variable
$x^i_r$ appearing in the subformula $\phi_\mathit{steady}(\sigma_i,\sigma'_i)$ and 
there is also an existential variable $y^i_r$ appearing in the subformula
$\phi_\mathit{step}(\sigma'_i,\sigma_{i+1})$.
We introduce a new existential variable $sum_r$ and enforce the constraint 
$sum_r = \sum_{i=0}^K x^i_r + \sum_{i=0}^{K-1} y^i_r$.
Using Lemma~\ref{lem:observation} we can easily prove that

\begin{lemma}\label{lem:another-observation}
	The following are true:
	\begin{itemize}
		\item Suppose $\phi_{reach}(\sigma,\sigma')$ is true with the 
		assignment $\{z_r\}_{r \in \ruleset}$ to the existential variables
		$\{sum_r\}_{r \in \ruleset}$. Then there is a schedule $\tau$
		between $\sigma$ and $\sigma'$ such that $\tau(r) = z_r$ for all
		rules $r$.
		\item Conversely, if $\tau$ is a schedule between $\sigma$ 
		and $\sigma'$, then it is possible to satisfy $\phi_{reach}(\sigma,\sigma')$ 
		by setting $\{\tau(r)\}_{r \in \ruleset}$
		to the variables in $\{sum_r\}_{r \in \ruleset}$.
	\end{itemize}
\end{lemma}

For liveness properties, this lemma will prove very useful.

\subsection{Proof of Corollary~\ref{cor:reach}}

\begin{proof}
	Recall that, in the parameterized reachability problem, 
	we are given two sets $\local_{=0},  \local_{>0}$ of locations, and we want to decide if there is an initial configuration $\gst_0$, and some configuration $\gst$ reachable from $\gst_0$ such that it satisfies $\gst.\counters[\ell] = 0$ for every 
	$\ell \in \local_{=0}$ and $\gst.\counters[\ell] > 0$ for every $\ell \in \local_{>0}$. The non-parameterized reachability problem
	is similar, except the initial configuration $\gst_0$ is also
	given as part of the input.
	
	By Theorem~\ref{th:main} we know that there is a formula $\phi_\mathit{reach}$ in 
	existential Presburger arithmetic with $(2\numlocal+2\numglob+2|\paraset|)$ free variables such that
	$\phi_\mathit{reach}(\sigma,\sigma')$ is true iff there is a run between $\sigma$ and
	$\sigma'$. 
	For parameterized reachability, 
	we modify the formula $\phi_{reach}$ so that the free variables become existential variables.
	For non-parameterized reachability, we modify $\phi_{reach}$ so that the free variables corresponding to the first configuration is fixed to $\sigma_0$ and the free variables corresponding to the second configuration become existential variables.
	Finally, in both the cases, we specify that the second configuration
	must satisfy the constraints according to the sets $\local_{=0}$
	and $\local_{>0}$.
	It is then clear that by checking if the constructed formula is 
	satisfiable, we can
	solve both the problems in NP.
\end{proof}

\section{Detailed proofs from Section~\ref{sec:liveness}}

We break down the proof of Theorem~\ref{th:liveness} into various parts.
As a preliminary step, we introduce the notion of a lasso path.

\subsection{Lasso paths and cut graphs}

Similar to model checking finite systems, we first show that
if a path of the threshold automaton satisfies the $\ELTLFT$ formula
$\varphi$, then there is an ``ultimately periodic'' path which 
also satisfies $\varphi$. \\

For a schedule $\tau$ let $[\tau]$ denote the set of all rules which
appear in $\tau$. We have the following definition:

\begin{definition} 
	A path $\finpath{\sigma}{\tau'}$ is called a lasso path
	if $\tau'$ can be decomposed as $\rho \circ \tau^{\omega}$ such that 
	if we let
	$\sigma_1 = \rho(\sigma)$ and $\sigma_2 = (\rho \circ \tau)(\sigma)$
	then
	\begin{itemize}
		\item $\sigma_1.\counters = \sigma_2.\counters$.
		\item If $r \in [\tau]$ such that $r.\update[x] > 0$
		then for all rules $r' \in [\tau]$, $x$ does not appear in any of 
		the fall guards of $r'$.
	\end{itemize}
\end{definition}

\begin{proposition} \label{prop:lasso}
	If $\infpath{\sigma'}{\tau'}$ satisfies a formula $\varphi$ in $\ELTLFT$
	then there is a lasso path  which satisfies 
	$\varphi$.
\end{proposition}

\begin{proof}
	We do not give details on Buchi automata and product construction, 
	as these are well-known.
	
	The formula $\varphi$ in $\ELTLFT$ can be thought of as an LTL
	formula over the atomic propositions $\mathit{pf}$ and hence
	there is a Buchi automaton $B_\varphi = (Q,2^{\mathit{pf}},Q_0,\Delta,F)$
	which recognizes exactly those sequence of propositions which
	satisfy $\varphi$.
	
	Fix a threshold automaton $\TA = (\local,\initlocal,\varset,\ruleset)$. Let $\Sigma$ be the set 
	of all configurations of $\TA$, $\iconfigs$ be the set
	of all initial configurations of $\TA$ and $\transrel \subseteq \Sigma \times \ruleset \times \Sigma$ 
	be the transition relation between the configurations.
	Let $\Sys(\TA) = (\Sigma,\iconfigs,\transrel)$.
	We can then construct
	the standard product construction $\Sys(\TA) \times B_{\varphi}$,
	where there is a transition $((\sigma,q),p,(\sigma',q'))$
	iff $(q,p,q') \in \Delta, (\sigma,r,\sigma') \in \transrel$ for some rule $r$
	and $\sigma \models p$. A path $(\sigma_0,q_0),(\sigma_1,q_1),\dots$
	is an accepting path iff it visits a state in $\Sigma \times F$ infinitely often.
	
	Let $PA = (\sigma_0,q_0),(\sigma_1,q_1),\dots$ be an infinite accepting path in $\Sys(\TA) \times B_\varphi$. This gives rise to a path 
	$P = \sigma_0,r_0,\sigma_1,r_1,\dots$ in $\Sys(\TA)$.
	Let $V \subseteq \globset$ be the set of all variables which
	are infinitely often incremented along $P$.
	Since these variables are infinitely often incremented,
	it follows that there exists $i$ such that for all $j \ge i$,
	$r_j$ does not have any fall guards with any of the variables in $V$.
	Since the number of processes does not change during the course of a 
	path and since the number of states in $F$ is finite, by combining all of the above facts, 
	it follows that we can find an infinite subsequence of $PA$ of the 
	form $(\sigma_{i_1},q_{i_1}),(\sigma_{i_2},q_{i_2}),\dots$ such that
	\begin{itemize}
		\item $q = q_{i_1} = q_{i_2} = \dots$ and $q \in F$
		\item $\sigma_{i_1}.\counters = \sigma_{i_2}.\counters = \dots$
		\item For every $j$, $\sigma_{i_j}.\vars[x] < \sigma_{i_{j+1}}.\vars[x]$ iff $x \in V$
		\item For every $i_1 \le j \le i_2$, the rule $r_j$ does
		not have any fall guards with any of the variables in $V$.
	\end{itemize}
	
	Let $\rho = r_1 \circ r_2 \circ \dots \circ r_{i_1-1}$ 
	and let $\tau = r_{i_1} \circ r_{i_1+1} \circ \dots \circ r_{i_2-1}$.
	It can be verified (using all the points given above) that
	$\finpath{\sigma_0}{\rho \circ \tau^{\omega}}$ is a lasso
	path which satisfies $\varphi$.
\end{proof}

Hence it suffices to concentrate on lasso paths in the future. 
To prove our theorem, we also need the notion of a cut graph of a formula
$\varphi$.
However, for our purposes, it suffices to know that
the cut graph $Gr(\varphi)$ of a formula $\varphi$
is a directed acyclic graph which consists of two distinguished vertices $loop_{st}$
and $loop_{end}$ and every other vertex is labelled by a (normal) sub-formula
of $\varphi$. 

For a formula  $\varphi \equiv \phi_0 \land \ltlF \phi_1 \land  \dots \land \ltlF \phi_k \land \ltlG \phi_{k+1}$ in normal form, let
$\mathtt{prop}(\varphi) = \phi_0$.
The following lemma (proved in~\cite{ELTLFT}) connects the notion of a lasso path and 
the cut graph. Intuitively, it says that a lasso path satisfies a formula
$\varphi$ iff it can be ``cut'' into finitely many pieces
such that (1) the endpoints of each piece satisfy some (propositional)
formula dictated by the cut-graph and (2) \emph{all} the configurations between any two endpoints satisfy some other (propositional) formula
as dictated by the cut-graph.

\begin{lemma}\label{lem:already}
	A lasso path $\infpath{\sigma_0}{\rho \circ \tau^{\omega}}$ satisfies
	a formula $\varphi \equiv \phi_0 \land \ltlF \phi_1 \land  \dots \land \ltlF \phi_k \land \ltlG \phi_{k+1}$ iff 
	there is a \emph{topological ordering}
	\begin{equation*}
		v_1, v_2, \dots, v_{c-1}, v_c = loop_{st}, v_{c+1}, \dots,
		v_{l-1}, v_l = loop_{end}
	\end{equation*}
	of its cut graph and a finite path
	\begin{equation*}
		\sigma_0, \tau_0, \sigma_1, \tau_1, \dots \sigma_c, \tau_c, \dots \sigma_{l-1}, \tau_{l-1}, \sigma_{l}
	\end{equation*}
	such that the following holds:
	If each $v_i$ (other than $v_c$ and $v_l$) is 
	of the form $\phi_0^i \land \ltlF \phi_1^i \land \dots \land \ltlF \phi_{k_i}^i \land \ltlG \phi_{k_i+1}^i$, then 
	
	\begin{itemize}
		\item $\sigma_0$ is an initial configuration,
		and $\tau_i(\sigma_i) = \sigma_{i+1}$ for every $0\le i < l$.
		\item $\tau_0 \circ \tau_1 \dots \circ \tau_{c-1} = \rho \circ\tau^{K}$ for some $K$
		and $\tau_c \circ \dots \circ \tau_{l-1} = \tau$ 
		\item $\sigma_0 \vDash \phi_0$ and $\setconf{\sigma_0}{\tau_0} \vDash \mathtt{prop}(\phi_{k+1})$.
		\item For every $i \notin \{c,l\}$, we have $\sigma_i \vDash \phi_0^i$.
		\item If $i < c$, then $\setconf{\sigma_i}{\tau_i} \vDash
		\bigwedge_{0 \le j \le i} \ \mathtt{prop}(\phi_{k_j+1}^j)$.
		\item If $i \ge c$, then $\setconf{\sigma_i}{\tau_i} \vDash \bigwedge_{0 \le j < l} \ \mathtt{prop}(\phi_{k_j+1}^j)$.
	\end{itemize}
\end{lemma}

We note the subtle but important difference in the indices between the 
last two points.
When $i < c$, we require $\setconf{\bf{\sigma_i}}{\tau_i} \vDash \bf{\bigwedge_{0 \le j \le i} \ \mathtt{prop}(\phi_{k_j+1}^j)}$, 
but when $i \ge c$, we require $\bf{\setconf{\sigma_i}{\tau_i} \vDash \bigwedge_{0 \le j < l} \ \mathtt{prop}(\phi_{k_j+1}^j)}$.\\

\textit{Remark: } In \cite{ELTLFT} the above lemma was only
proven for automata where there are no updates in a cycle,
but virtually the same proof also holds for the general case.
\medskip

Combining Proposition~\ref{prop:lasso} and Lemma~\ref{lem:already},
we get the following lemma.

\begin{corollary}\label{cor:topo-ord}
	The formula $\varphi \equiv \phi_0 \land \ltlF \phi_1 \land  \dots \land \ltlF \phi_k \land \ltlG \phi_{k+1}$
	is satisfiable if and only if there is a \emph{topological ordering} $<_{Gr}$
	\begin{equation*}
		v_1,v_2,\dots, v_{c-1}, v_c = loop_{st}, v_{c+1}, \dots,
		v_{l-1}, v_l = loop_{end}
	\end{equation*}
	of its cut graph and a finite path $P$
	\begin{equation*}
		\sigma_0, \tau_0, \sigma_1, \tau_1, \dots \sigma_c, \tau_c, \dots \sigma_{l-1}, \tau_{l-1}, \sigma_{l}
	\end{equation*}
	such that the following holds:
	If each $v_i$ (other than $v_c$ and $v_l$) is 
	of the form $\phi_0^i \land \ltlF \phi_1^i \land \dots \land \ltlF \phi_{k_i}^i \land \ltlG \phi_{k_i+1}^i$, then 
	\begin{itemize}
		\item $\sigma_0$ is an initial configuration, $\sigma_c.\counters = \sigma_l.\counters$ and 
		$\tau_i(\sigma_i) = \sigma_{i+1}$ for every $i < l$.
		\item If $r \in [\tau_c \circ \tau_{c+1} \dots \circ \tau_{l-1}]$ such that 
		$r.\update[x] > 0$ then for every rule $r' \in [\tau_c \circ \tau_{c+1} \dots \circ \tau_{l-1}]$, 
		$x$ does not appear in any of the fall guards of $r'$.
		\item $\sigma_0 \vDash \phi_0$ and $\setconf{\sigma_0}{\tau_0} \vDash \mathtt{prop}(\phi_{k+1})$.
		\item For every $i \notin \{c,l\}$, we have $\sigma_i \vDash \phi_0^i$.
		\item If $i < c$, then $\setconf{\sigma_i}{\tau_i} \vDash
		\bigwedge_{0 \le j \le i} \mathtt{prop}(\phi_{k_j+1}^j)$.
		\item If $i \ge c$, then $\setconf{\sigma_i}{\tau_i} \vDash \bigwedge_{0 \le j < l} \mathtt{prop}(\phi_{k_j+1}^j)$.
	\end{itemize}	
	If such a topological ordering $<_{Gr}$ and such a path 
	$P$ exists then we call the pair $(<_{Gr},P)$ a witness to the 
	formula $\varphi$.
\end{corollary}


This finishes the first part of the proof of theorem~\ref{th:liveness}.

\subsection{Roadmap}

Before we move on to the next part, we give some intuition.
Notice that for $(<_{Gr},P)$ to be a witness for $\varphi$, 
among other restrictions,
it has to satisfy some reachability conditions
($\sigma_{i+1} = \tau_i(\sigma_i)$), some safety conditions ($\sigma_i \vDash \phi$) 
as well as some liveness conditions ($\setconf{\sigma_i} {\tau_i} \vDash \phi$).
Our strategy for the next part is the following:
Intuitively, the main hindrance for us to apply 
our main result (that the reachability relation is existential Presburger definable) here are the liveness conditions. 
Here is where the multiplicativity condition comes into picture.
For a propositional formula $p$, we will
define an existential Presburger arithmetic formula 
$\phi_p$ with the following properties:
\begin{itemize}
	\item If there is a path from $\sigma$ to $\sigma'$
	such that all the configurations in the path satisfy $p$
	then $\phi_p(\sigma,\sigma')$ will be true.
	\item If $\phi_p(\sigma,\sigma')$ is true then there is 
	a path from $2 \cdot \sigma$ to $2 \cdot \sigma'$ such that
	all the configurations in the path satisfy $p$.
\end{itemize}

Using this newly defined formula $\phi_p$ and some additional tricks,
we then show that for a topological ordering $<_{Gr}$, we can write a formula $\phi_{live}$ such that
$\phi_{live}$ will be satisfied iff there is a path $P$ 
such that $(<_{Gr},P)$ is a witness to the given specification.


\subsection{The usefulness of multiplicativity}
First, given a path $P$ from $\sigma_0$ to $\sigma_m$, we define
a lifted path $\mathtt{lift}(P)$ between $2 \cdot \sigma$ and
$2 \cdot \sigma_m$.

\begin{definition}
	Suppose $\sigma_0, \tau_0, \sigma_0', t_0, \sigma_1, \tau_1, \sigma_1', t_1, \dots, \sigma_{m-1}, \tau_{m-1}, \sigma_{m-1}', t_{m-1}, \sigma_m$
	is a path such that each $\tau_i$ is a steady schedule between $\sigma_i$ and $\sigma_{i+1}$ and 
	each $t_i$ is a rule such that $\sigma_i'$ and $\sigma_{i+1}$ have different contexts. Its corresponding lifted path 
	$\mathtt{lift}(P)$ is defined by $2 \cdot \sigma_0, 2 \cdot \tau_0, 2 \cdot \sigma_0', 2 \cdot t_0, 2 \cdot \sigma_1, 2 \cdot \tau_1, \dots, 2 \cdot t_{m-1}, 2 \cdot \sigma_m$.
\end{definition}

Using multiplicativity it is easy to verify that $\mathtt{lift}(P)$
is a valid path from $2 \cdot \sigma_0$ to $2 \cdot \sigma_m$.

We say that a path $P = \sigma_0, \tau_0, \sigma_1, \dots, \sigma_m, \tau_m$
satisfies a proposition $p$ if all the configurations in $P$
satisfy $p$. We now prove the following lemma:

\begin{lemma}\label{lem:reach-useful}
	Let $p$ be a propositional $\mathit{pf}$-formula as given in $\ELTLFT$.
	Then there exists an existential Presburger formula $\phi_{p}$ such that
	\begin{itemize}
		\item If there is a path $P$ from $\sigma$ to $\sigma'$
		such that $P$ satisfies $p$
		then $\phi_p(\sigma,\sigma')$ will be true.
		\item If $\phi_p(\sigma,\sigma')$ is true then there is
		a path $P$ from $\sigma$ to $\sigma'$ such that the lifted path $\mathtt{lift}(P)$ 
		from $2 \cdot \sigma$ to $2 \cdot \sigma'$ satisfies $p$.
	\end{itemize}
\end{lemma}

\begin{proof}
	Recall that by theorem~\ref{th:main} that the reachability relation $\phi_{reach}(\sigma,\sigma')$ is
	\begin{equation*}
		\exists \ \sigma_0, \sigma'_0, \dots, \sigma_{K}, \sigma'_{K} \
		\left(\sigma_0 = \sigma \land \sigma'_{K} = \sigma' \land \bigwedge_{0 \le i \le K} \phi_\mathit{steady}(\sigma_i,\sigma'_i) 
		\land  \bigwedge_{0 \le i \le K-1}
		\phi_\mathit{step}(\sigma'_i,\sigma_{i+1}) \right)
	\end{equation*}
	
	where $K = |\Phi| + 1$, \ 
	$\phi_{steady}(\sigma_i,\sigma'_i)$ is true iff there is 
	a steady schedule between $\sigma_i$ and $\sigma'_i$
	and $\phi_{step}(\sigma'_i,\sigma_{i+1})$ is true iff $\sigma_{i+1}$
	can be reached from $\sigma'_i$ in at most one step.
	Further recall that for each rule $r$ there is 
	an existential variable $sum_r$ which denotes the 
	number of times the rule $r$ is fired in the run
	between $\sigma$ and $\sigma'$.
	
	Let $p$ be a $\mathit{pf}$-formula as given in $\ELTLFT$.
	Depending on the structure of $p$ we introduce a different formula
	for $\phi_p(\sigma,\sigma')$:
	In each case, it will be immediately clear that if there is a path $P$ 
	from $\sigma$ to $\sigma'$ satisfying $p$, then
	$\phi_p(\sigma,\sigma')$ is true. For this reason, we will only discuss 
	how to prove the other direction of the claim (namely, if $\phi_p(\sigma,\sigma')$ is true then there is
	a path $P$ from $\sigma$ to $\sigma'$ such that the lifted path $\mathtt{lift}(P)$ 
	from $2 \cdot \sigma$ to $2 \cdot \sigma'$ satisfies $p$.)
	\begin{itemize}
		\item $p := (S = 0)$ for some $S \subseteq \local$. 
		In this case, we construct $\phi_p$ from $\phi_{reach}$ as follows:
		To each of the formulas $\phi_\mathit{steady}(\sigma_i,\sigma'_i)$,
		we add the following constraint
		\begin{equation}
			\bigwedge_{\ell \in S} \sigma_i.\counters[\ell] = 0 \land  
			\bigwedge_{\ell \in S} \sigma'_i.\counters[\ell] = 0 \land  
			\bigwedge_{\ell \in S} \bigwedge_{r \in In_{\ell}} sum_r = 0
		\end{equation}
		where
		$In_{\ell}$ denotes the set of all incoming transitions to $\ell$.
		By using Lemma~\ref{lem:another-observation} it is clear that if
		$\phi_p(\sigma,\sigma')$ is true then there is a path
		$P$ between $\sigma$ and $\sigma'$ such that \emph{all}
		the configurations in the path $P$ satisfy $p$.
		Consequently, it is then easy to verify that 
		all the configurations in the lifted 
		path $\mathtt{lift}(P)$ also satisfy $p$.\\
		
		\item $p := \lnot (S = 0)$ for some $S \subseteq \local$. In this case, we construct $\phi_p$ from $\phi_{reach}$ as follows:
		To each of the formulas $\phi_\mathit{steady}(\sigma_i,\sigma'_i)$ 
		we add the following constraint
		\begin{equation}
		\bigvee_{\ell \in S} \sigma_i.\counters[\ell] > 0 \ \land \
		\bigvee_{\ell \in S} \sigma'_i.\counters[\ell] > 0
		\end{equation}
		Using Lemma~\ref{lem:another-observation} it is then clear that if
		$\phi_p(\sigma,\sigma')$ is true then there exists a path
		between $\sigma$ and $\sigma'$ such that each of the \emph{intermediate} configurations $\sigma_0,\sigma_0',\dots,\sigma_K,\sigma_K'$ satisfy $p$. 
		By a careful inspection of the corresponding lifted path
		$\mathtt{lift}(P)$ it can be verified that \emph{every} configuration in the lifted path satisfies $p$.
		
		\item $p := (S' = 0) \land \bigwedge_{S \in T} \lnot (S = 0)$.
		In this case, to each $\phi_{steady}(\sigma_i,\sigma_i')$
		we add 
		\begin{multline} \label{eq:conjunct}
		\bigwedge_{\ell \in S'} \sigma_i.\counters[\ell] = 0 \ \land \
		\bigwedge_{\ell \in S'} \sigma'_i.\counters[\ell] = 0 \ \land \
		\bigwedge_{\ell \in S'} \bigwedge_{r \in In_{\ell}} sum_r = 0 \ \land \\
		\left(\bigwedge_{S \in T} \bigvee_{\ell \in S} \sigma_i.\counters[\ell] > 0\right) \ \land \
		\left(\bigwedge_{S \in T} \bigvee_{\ell \in S} \sigma'_i.\counters[\ell] > 0\right)
		\end{multline}

		\item $p := (g \implies c)$ where $g$ is a $\mathit{gf}$-formula and $c:= (S' = 0) \land \bigwedge_{S \in T} \lnot (S = 0)$ is a $\mathit{cf}$-formula. 
		In this case, to each $\phi_{steady}(\sigma_i,\sigma_i')$ we
		add the constraint 
		\begin{equation}
			(\sigma_i \models g) \implies  \phi_c
		\end{equation}
		where $\phi_c$ is the constraint given by equation~(\ref{eq:conjunct}).
	\end{itemize}
	
\end{proof}

Given a path $P$, let $P(r)$ denote the number of times
the rule $r$ appears in $P$.
Notice that once again we have proved something stronger.

\begin{lemma}\label{lem:observation-prop}
	Let $p$ be a propositional formula in $\ELTLFT$. Then,
	\begin{itemize}
		\item If there is a path $P$ between $\sigma$ and $\sigma'$
		satisfying $p$ then it is possible to satisfy $\phi_p(\sigma,\sigma')$ by setting $\{P(r)\}_{r \in \ruleset}$
		to the variables $\{sum_r\}_{r \in \ruleset}$.
		\item Suppose $\phi_p(\sigma,\sigma')$ is satisfiable.
		Then there is path $P$ between $\sigma$ and $\sigma'$
		such that $P(r) = sum_r$ for every rule $r$ and 
		the lifted path $\mathtt{lift}(P)$ satisfies $p$.
	\end{itemize}
\end{lemma}

\subsection{Proof of Theorem~\ref{th:liveness}}

Using the above lemma, we now give a proof of Theorem~\ref{th:liveness}.
The proof proceeds by encoding the conditions mentioned in Corollary~\ref{cor:topo-ord}.

Let $\varphi \equiv \phi_0^0 \land \ltlF \phi_1^0 \land \dots \land \ltlF \phi_{k_0}^0 \land \ltlG \phi_{k_0+1}^0$. 
We want to guess a witness $(<_{Gr},P)$. For this reason,
we first guess and fix a topological ordering $<_{Gr}$
\begin{equation*}
	v_1, v_2, \dots, v_c = loop_{st}, v_{c+1}, \dots, v_l = loop_{end}
\end{equation*}
of the cut graph $Gr(\varphi)$. Let each $v_i$ (apart from $v_c$ and $v_l$)
be of the form $\phi_0^i \land \ltlF \phi_1^i \land \dots \land \ltlF \phi_{k_i}^i \land \ltlG \phi_{k_i+1}^i$.

We will write a formula $\phi_{live}$ with existential variables
decribing $l+1$ configurations $\sigma_0,\sigma_1,\dots,\sigma_c,\dots, \sigma_l$.
We will now show to enforce each of the conditions mentioned by 
Corollary~\ref{cor:topo-ord} as a formula in 
existential Presburger arithmetic. (For ease of presentation,
we provide the steps in an order different from the one in Corollary~\ref{cor:topo-ord})

\paragraph{First condition: } $\sigma_0$ is an initial configuration
and $\sigma_c.\counters = \sigma_l.\counters$: Clearly this can
be encoded as an existential Presbuger formula.

\paragraph{Second condition: } For each $i \notin \{c,l\}, \ \sigma_i \vDash \phi_0^i$ : Since each $\phi^0_i$ is a propositional $\mathit{pf}$ formula, this constraint can once again be encoded as an existential Presbuger formula.

\paragraph{Third condition: } For each $i < c$, the path
between $\sigma_i$ and $\sigma_{i+1}$ must satisfy $\bigwedge_{0 \le j \le i} \ \mathtt{prop}(\phi^j_{k_j+1})$ : Let $p_i := \bigwedge_{0 \le j \le i} \ \mathtt{prop}(\phi^j_{k_j+1})$. We enforce this condition by putting
the constraint
\begin{equation}\label{eq:first}
	\phi_{p_i}(\sigma_i,\sigma_{i+1})
\end{equation} 

\paragraph{Fourth condition: } For each $i \ge c$, the path
between $\sigma_{i}$ and $\sigma_{i+1}$ must satisfy $\bigwedge_{0 \le j \le l} \ \mathtt{prop}(\phi^j_{k_j+1})$ : Let $p_{i} := \bigwedge_{0 \le j \le l} \ \mathtt{prop}(\phi^j_{k_j+1})$. We enforce this condition by putting
the constraint
\begin{equation}\label{eq:second}
	\phi_{p_{i}}(\sigma_{i},\sigma_{i+1})
\end{equation} 

\paragraph{Fifth condition: } If $r$ is a rule which is fired in 
the schedule between $\sigma_c$ and $\sigma_l$ such that $r.\vec{u}[x] > 0$
for some shared variable $x$, then for every rule $r'$ which is fired 
in the schedule between $\sigma_c$ and $\sigma_l$, $x$ should not 
appear in any of the fall guards of $r'$ : 
Recall that for every $0 \le i < l$ and for every rule $r$,
the formula $\phi_{p_i}(\sigma_i,\sigma_{i+1})$ (as defined
in ~\ref{eq:first} and ~\ref{eq:second}) has an existential variable
$sum^i_r$ which denotes the number of times the rule $r$ has 
been fired in the schedule between $\sigma_i$ and $\sigma_{i+1}$. 
Hence to enforce this  condition, we put the following constraint:

\begin{equation*}
	\bigwedge_{r, r' \in \ruleset, x \in \varset} \left(\left(\left(\bigvee_{c \le i < l} sum^i_r > 0 \right) \land r.\vec{u}[x] > 0 \land r' \in Fall_x\right) \implies \left(\bigwedge_{c \le i < l} sum^i_{r'} = 0\right)\right)
\end{equation*}

Here $Fall_x$ denote the set of all rules which have a fall guard with 
the variable $x$.

Let $\phi_{live}$ be the above formula. 
From the construction of the formula and Lemma~\ref{lem:observation-prop}
it is clear that if $(<_{Gr},P)$ is a witness then 
$\phi_{live}$ is satisfiable.
On the other hand, suppose $\phi_{live}$ is satisfiable.
We therefore get a sequence of configurations $\eta_0,\dots,\eta_l$
satisfying the formula $\phi_{live}$. 
By Lemma~\ref{lem:observation-prop} we get that for every $i$,
there is a path $P_i$ between $\eta_i$ and $\eta_{i+1}$ such that
$P_i(r) = sum^i_r$ for every rule $r$ and the lifted path $\mathtt{lift}(P_i)$
satisfies the proposition $p_i$. It is then easy to see that
the pair $(<_{Gr},Q)$ where $Q$ is given by,

\begin{equation*}
	\mathtt{lift}(P_0), \mathtt{lift}(P_1), \dots, \mathtt{lift}(P_{l-1})
\end{equation*}

is a witness.\\

It then follows that model-checking formulas in $\ELTLFT$ is in NP.

\section{Detailed proofs of Section~\ref{sec:synthesis}}

\subsection*{Proof of theorem~\ref{th:bound-synthesis-upperbound}}

\begin{proof}
	Given a sketch threshold automaton $\TA$, formula $\varphi$, and polynomial $p$, we
	first guess in polynomial time an assignment $\mu$.
	Clearly we can check if $\mu$ makes the environment multiplicative in NP
	as the conditions of definition~\ref{def:mult} can be encoded 
	as a(n) (integer) linear program.
	Using Theorem~\ref{th:liveness},
	we can decide if all executions of $\TA[\mu]$ 
	satisfy $\lnot \varphi$ in co-non-deterministic 
	polynomial time in $\TA[\mu]$. So we get a $\Sigma_2^p$ procedure
	for the bounded synthesis problem.
\end{proof}

\subsection*{Proof of theorem~\ref{th:bound-synthesis-lowerbound}}

The upper bound follows from Theorem~\ref{th:bound-synthesis-upperbound} and
so we only prove the lower bound here, by giving a reduction from the
$\Sigma_2$-3-SAT problem, which is well-known to be $\Sigma_2^p$-hard.

Let $\exists x_1,\dots,x_m  \ \forall y_1,\dots,y_k \ 
\phi(x,y)$ be a $\Sigma_2$-3-SAT formula where $\phi(x,y)$ is in disjunctive normal form. 
We construct a threshold automaton as follows: (See Figure~\ref{fig:example} for an illustrative example)

Let $Env = (\paraset, \ResCond,\syssize)$ be the environment
where $\paraset$ consists of a single variable $n$, 
$\ResCond$ is simply $n \ge 1$ and $\syssize$ is the identity function.

The sketch threshold automaton $\TA$ is constructed as follows:
For every variable $x_i$ we will have a location $\ell_{x_i}$. Further
we will also have a start location $\ell_{x_0}$.
Let $a, b_1,\dots,b_m,\bar{b_1},\dots,\bar{b_m}$ be shared variables. 
Our construction will ensure that $a$ will never be incremented.
Let $v_1,\dots,v_m$ be indeterminates which will appear in the 
threshold guards.
Between $\ell_{x_{i-1}}$
and $\ell_{x_i}$ there are two rules of the form $(\ell_{x_{i-1}},\ell_{x_i}, a < v_i \cdot n, \cpp{{b_i}})$ 
and $(\ell_{x_{i-1}},\ell_{x_i}, a \ge v_i \cdot n, \cpp{\bar{b_i}})$.
Since $a$ will never be incremented (and hence always stay $0$),
it follows that depending on the value that $v_i$ is assigned to, exactly
one of these two rules can be fired.

For every variable $y_i$ we will have a location $\ell_{y_i}$ and
two other locations $\ell_{z_i}, \ell_{\bar{z_i}}$.  Further we will also have a location $\ell_{y_0}$ and a transition
$(\ell_{x_m},\ell_{y_0},\true,0)$. 
Let $c_1,\dots,c_k,\bar{c_1},\dots,\bar{c_k}$ be shared variables.
From $\ell_{y_{i-1}}$ there are two rules, 
$(\ell_{y_{i-1}},\ell_{z_i},true,\cpp{c_i})$ and
$(\ell_{y_{i-1}},\ell_{\bar{z_i}},true,\cpp{\bar{c_i}})$. 
Intuitively, firing the first rule corresponds to making $y_i$
to be true and firing the second rule corresponds to making it false.
We will ensure that if two different processes fire two different rules
from $\ell_{y_{i-1}}$ then all the processes get stuck:
We add two rules $(\ell_{z_i},\ell_{y_i},c_i \ge n,\vec{0})$
and $(\ell_{\bar{z_i}},\ell_{y_i},\bar{c_i} \ge n,\vec{0})$.
These two rules ensure that if two different processes fired
two different rules from the location $\ell_{y_{i-1}}$ then they
all get stuck and cannot progress to $\ell_{y_i}$.

For a variable $x_i$, let $var(x_i) = b_i$ and $var(\bar{x_i}) = \bar{b_i}$.
Similarly for a variable $y_i$, let $var(y_i) = c_i$ and $var(\bar{y_i}) = \bar{c_i}$.
Now for every disjunct $D = p \land q \land r$ that appears in $\phi(x,y)$
we have the following transition from $\ell_{y_k}$ to a new location $\ell_{F}$: $(\ell_{y_k},\ell_F,var(p) \ge 1 \land var(q) \ge 1 \land var(r) \ge 1,0)$. 
Finally both $\ell_{y_k}$ and $\ell_F$ have self-loops with
no guards.

Let the above sketch threshold automaton be $\TA$.
Let $\varphi'$ be the following formula:
\begin{equation*}
	\left(\bigvee_{(p \land q \land r)  \in \phi(x,y)} var(p) \ge 1 \land var(q) \ge 1 \land var(r) \ge 1 \right)
\end{equation*}
and let $\varphi$ be the following formula:
$(\ltlF \ltlG (\varphi' \implies \ell_{y_k} = 0)) \land (\ltlG \ell_{F} = 0)$.
The first part of $\varphi$ says that if any one of the rules
between $\ell_{y_k}$ and $\ell_F$ get unlocked at some point,
then all the processes leave $\ell_{y_k}$ (and subsequently go
to $\ell_F$). The second part of $\varphi$ says that no process
ever reaches $\ell_F$. Hence $\varphi$ can be satisfied by a path
iff there is a process in $\ell_{y_k}$ but none of the outgoing rules
from $\ell_{y_k}$ are unlocked.

We notice the following observation about the constructed threshold automaton: Let $\mu$ be any assignment to the indeterminates.
In any infinite run of $\TA[\mu]$ at least one process reaches
the location $\ell_{y_k}$. Further, in any run in which some process reaches the location $\ell_{y_k}$, the run satisfies the following:
For every $1 \le i \le m$, either all configurations of the run
have $b_i = 0$ or all configurations of the run have $\bar{b_i} = 0$.
Similarly for every $1 \le j \le k$, either all configurations
of the run have $c_i = 0$ or all configurations of the run
have $\bar{c_i} = 0$.

Let $p$ be the identity function.
We now show that there exists an assignment to the indeterminates
$\{v_i\}$ (using $p(|\TA|+|\varphi|)$ many bits)
such that $\lnot \varphi$ is satisfied by \emph{every} infinite run iff the formula $\exists x_1,\dots,x_m \ \forall y_1,\dots,y_n \ \phi(x,y)$
is true.

Suppose there exists an assignment $\mu$ to the indeterminates $\{v_i\}$
such that $\lnot \varphi$ is true for every infinite run starting
from every initial configuration in the threshold automaton $\TA[\mu]$. Consider the following assignment $\nu$
for the variables $x_i$: $\nu(x_i) = \true$ if $\mu(v_i) > 0$
and $\nu(x_i) = \false$ if $\mu(v_i) = 0$. It is now an easy check
that the formula $\forall y_1 \dots y_n \ \phi(\nu(x),y)$ is true.

Suppose the formula $\exists x_1,\dots,x_m \ \forall y_1,\dots,y_n \ \phi(x,y)$ is true. Let $\nu$ be an assignment to the variables $\{x_i\}$
such that $\forall y_1 \dots y_n \ \phi(\nu(x),y)$ is true.
Consider an assignment $\mu$ to the indeterminates $\{v_i\}$ as follows:
$\mu(v_i) = 1$ if $\nu(x_i) = \true$ and $\mu(v_i) = 0$ if $\nu(x_i) = \false$. 
It is once again easy to check that $\lnot \varphi$ is true
for every infinite run starting
from every initial configuration in the threshold automaton $\TA[\mu]$.
  
\begin{figure} \label{fig:example}
	\tikzstyle{node}=[circle,draw=black,thick,minimum size=5mm,inner sep=0.75mm,font=\normalsize]
	\tikzstyle{edgelabelabove}=[sloped, above, align= center]
	\tikzstyle{edgelabelbelow}=[sloped, below, align= center]
		\begin{tikzpicture}[->,node distance = 1.5cm,scale=0.8, every node/.style={scale=0.8}]
		\node[node] (x0) {$\ell_{x_0}$}; 
		\node[node, right = of x0] (x1) {$\ell_{x_1}$};
		\node[node, right = of x1] (x2) {$\ell_{x_2}$};
		
		\draw(x0) edge[edgelabelabove, bend left] node{${\tiny a < v_1 \cdot n \mapsto \cpp{{b_1}}}$} (x1);
		\draw(x0) edge[edgelabelbelow, bend right] node{$a \ge v_1 \cdot n\mapsto \cpp{\bar{b_1}}$} (x1);
		\draw(x1) edge[edgelabelabove, bend left] node{$a < v_2 \cdot n \mapsto \cpp{b_2}$} (x2);
		\draw(x1) edge[edgelabelbelow, bend right] node{$a \ge v_2  \cdot n\mapsto \cpp{\bar{b_2}}$} (x2);
		
		\node[node, right = of x2, xshift = -0.5cm] (y0) {$\ell_{y_0}$};
		\node[node, above right = of y0, xshift = -0.5cm] (z1) {$\ell_{z_1}$};
		\node[node, below right = of y0, xshift = -0.5cm] (bz1) {$\ell_{\bar{z_1}}$};
		\node[node, right = of y0] (y1) {$\ell_{y_1}$};
		\node[node, right = of y1, xshift = 1.5cm] (F) {$\ell_{F}$};
		
		\draw(x2) edge[edgelabelabove] node{$\true$} (y0);
		
		\draw(y0) edge[edgelabelabove, bend left] node{$\cpp{c_1}$} (z1);
		\draw(y0) edge[edgelabelbelow, bend right] node{$\cpp{\bar{c_1}}$} (bz1);
		
		\draw(z1) edge[edgelabelbelow, bend left] node{$c_1 \ge n$} (y1);
		\draw(bz1) edge[edgelabelabove, bend right] node{$\bar{c_1} \ge n$} (y1);

		\draw(y1) edge[edgelabelabove, bend left] node{$b_1 \ge 1 \land c_1 \ge 1 \land \bar{c_1} \ge 1$} (F);
		\draw(y1) edge[edgelabelbelow, bend right] node{$\bar{b_1} \ge 1 \land \bar{b_2} \ge 1 \land c_1 \ge 1$} (F);
		
		\draw(y1) edge[edgelabelabove, loop right] (y1);
		\draw(F) edge[edgelabelabove, loop above] (F);
		\end{tikzpicture}
		\caption{Threshold automaton corresponding to the formula
			$\exists x_1, x_2 \ \forall y_1 \ (x_1 \land y_1 \land \bar{y_1}) \lor (\bar{x_1} \land \bar{x_2} \land y_1)$}
\end{figure}

\end{document}